\documentclass[twoside]{IEEEtran}
\usepackage{amsmath}
\usepackage{amssymb}
\usepackage{amsthm}
\usepackage{enumerate}

\usepackage{epsfig}
\usepackage{subfigure}
\usepackage{psfrag}
\usepackage{cite}
\usepackage{stfloats}
\usepackage[mathscr]{euscript}
\usepackage{hyperref}
\hypersetup{backref, colorlinks=true, linkcolor=blue}

\newtheorem{thm}{Theorem}
\newtheorem{cor}[thm]{Corollary}
\newtheorem{lemma}{Lemma}

\newtheorem{defn}{Definition}

\theoremstyle{remark}
\newtheorem{remark}{Remark}
\newtheorem{example}{Example}

\newcommand{\bigo}[1]{O\left(#1\right)}
\newcommand{\smallo}[1]{o\left(#1\right)}
\newcommand{\Qinv}[1]{Q^{-1}\left(#1\right)}
\newcommand{\Var}[1]{{\rm Var}\left[#1\right]}
\newcommand{\E}[1]{\mathbb E\left[#1\right]}
\newcommand{\Prob}[1]{\mathbb P\left[#1\right]}

\newcommand{\beq}{\begin{equation}}
\newcommand{\eeq}{\end{equation}}

\graphicspath{{figures/}}

\psfrag{n}{\small{$n$}}
\psfrag{R}{\small{$R$}}
\psfrag{Acost}{\small{Achievability \eqref{eq:DT}}}
\psfrag{Ccost}{\small{Converse \eqref{eq:Ccosta}}}
\psfrag{2orderCost}{\small{Approximation \eqref{eq:2ordercost}}}
\psfrag{Capacitycost}{\small{$C(\beta)$}}

\newif\ifmapx
{\catcode`/=0 \catcode`\\=12/gdef/mkillslash\#1{#1}}
\edef\jobnametmp{\expandafter\string\csname varrate_apx\endcsname}
\edef\jobnameapx{\expandafter\mkillslash\jobnametmp}
\edef\jobnameexpand{\jobname}
\ifx\jobnameexpand\jobnameapx
\mapxtrue
\else
\mapxfalse
\fi

\long\def\apxonly#1{\ifmapx{\color{blue}#1}\fi}

\title{Channels with cost constraints: \\strong converse and dispersion}

\author{%
{\small\begin{minipage}{\linewidth}\begin{center}
\begin{tabular}{ccc}
{\large{Victoria Kostina}} 
&&
{\large{Sergio Verd\'{u}}}\\
California Institute of Technology 
&&
Princeton University \\
Pasadena, CA 91125 
&&
Princeton, NJ 08544\\
\url{vkostina@caltech.edu} 
&&
\url{verdu@princeton.edu}
\end{tabular}
\end{center}\end{minipage}}
\thanks{
This work was supported in part by the National Science Foundation (NSF)
under Grant CCF-1016625 and by the Center for Science of Information
(CSoI), an NSF Science and Technology Center, under Grant CCF-0939370.
}
}

\IEEEoverridecommandlockouts
\begin{document}

\maketitle
\begin{abstract}
This paper shows the strong converse and the dispersion of memoryless channels with cost constraints and performs refined analysis of the third order term in the asymptotic expansion of the maximum achievable channel coding rate, showing that it is equal to $\frac 1 2 \frac {\log n}{n}$ in most cases of interest. The analysis is based on a non-asymptotic converse bound expressed in terms of the distribution of a random variable termed the $\mathsf b$-tilted information density, which plays a role similar to that of the $\mathsf d$-tilted information in lossy source coding. We also analyze the fundamental limits of lossy joint-source-channel coding over channels with cost constraints. 
\end{abstract}


\section{Introduction}
\label{sec:intro}
This paper is concerned with the maximum channel coding rate achievable at average error probability $\epsilon > 0$ where the cost of each codeword is constrained.
The capacity-cost function $C(\beta)$ of a channel specifies the maximum achievable channel coding rate compatible with vanishing error probability and with codeword cost not exceeding $\beta$ in the limit of large blocklengths. 

A channel is said to satisfy the strong converse if $\epsilon \to 1$ as $n \to \infty$ for any code operating at a rate above the capacity. For memoryless channels without cost constraints, the strong converse was first shown by Wolfowitz: \cite{wolfowitz1957coding} treats the discrete memoryless channel (DMC), while  \cite{wolfowitz1959strong} generalizes the result to memoryless channels whose input alphabet is finite while the output alphabet is the real line. Arimoto \cite{arimoto1973converse} showed a new converse bound stated in terms of Gallager's random coding exponent, which also leads to the strong converse for the DMC. Dueck and K\"orner \cite{dueck1979reliability} found the reliability function of DMC for rates above capacity, a result which implies a strong converse. Kemperman \cite{kemperman1973strong} showed that the strong converse holds for a DMC with feedback. A simple proof of strong converse for memoryless channels that does not invoke measure concentration inequalities was recently given in \cite{PV10-ari}.
For a class of discrete channels with finite memory, the strong converse was shown by Wolfowitz \cite{wolfowitz1958maximum} and independently by Feinstein \cite{feinstein1959onthecodingtheorem}, a result soon generalized to a more general class of stationary discrete channels with finite memory \cite{Wolfowitz1960anoteonthestrongconverse}. In a more general setting not requiring the assumption of stationarity or finite memory, Verd\'u and Han \cite{verdu1994general} showed a necessary and sufficient condition for a channel without cost constraints to satisfy the strong converse, while Han \cite[Theorem 3.7.1]{han2003information} generalized that condition to the setting with cost constraints. In the special case of finite-input channels, that necessary and sufficient condition boils down to the capacity being equal to the limit of maximal normalized mutual informations. In turn, that condition is implied by the information stability of the channel \cite{pinsker1964information}, a condition which in general is not easy to verify. Using a novel notion of strong information stability, a general strong converse result was recently shown in \cite[Theorem 3]{polyanskiy2011relativeentropy}.
The strong converse for DMC with separable cost was shown by Csisz\'ar and K\"orner \cite[Theorem 6.11]{csiszar2011information} and by Han \cite[Theorem 3.7.2]{han2003information}. Regarding continuous channels, in the most basic case of the memoryless additive white Gaussian noise (AWGN) channel with the cost function being the power of the channel input block, $\mathsf b_n(x^n) = \frac 1 n |x^n|^2$, the strong converse was shown by Shannon \cite{shannon1959probability} (contemporaneously with Wolfowitz's finite-alphabet strong converse). Yoshihara \cite{yoshihara1964simple} proved the strong converse for the time-continuous channel with additive Gaussian noise having an arbitrary spectrum and also gave a simple proof of Shannon's strong converse result.  Under the requirement that the power of each message converges stochastically to a given constant $\beta$, the strong converse for the AWGN channel with feedback was shown by Wolfowitz\cite{Wolfowitz1968noteonthegaussian}. Note that in all those analyses of the power-constrained AWGN channel the cost constraint is meant on a per-codeword basis. In fact,  the strong converse ceases to hold if
the cost constraint is averaged over the codebook \cite[Section 4.3.3]{polyanskiy2010thesis}.  

 Channel dispersion quantifies the backoff from capacity, unescapable at finite blocklengths due to the random nature of the channel coming into play, as opposed to the asymptotic representation of the channel as a deterministic bit pipe of a given capacity. More specifically, for coding over the DMC, the maximum achievable code rate at blocklength $n$ compatible with error probability $\epsilon$ is approximated by 
$C - \sqrt{\frac V n} \Qinv{\epsilon} $ \cite{strassen1962asymptotische,polyanskiy2010channel}
where $C$ is the channel capacity, $V$ is the channel dispersion, and $\Qinv{\cdot}$ is the inverse of the Gaussian complementary cdf. 
 Polyanskiy et al. \cite{polyanskiy2010channel} found the dispersion of the DMC without cost constraints as well as that of the AWGN channel with a power constraint. In parallel, Hayashi \cite[Theorem 3]{hayashi2009secondorder} gave the dispersion of the DMC with and without cost constraints (with the loose  estimate of $\smallo{\sqrt n}$ for the third order term). For constant composition codes over the DMC, Polyanskiy \cite[Sec. 3.4.6]{polyanskiy2010thesis} showed the dispersion of constant composition codes over the DMC, 
  while Moulin \cite{moulin2012logvolume} refined the third-order term in the expansion of the maximum achievable code rate, under regularity conditions. Wang et al. \cite{wang2011dispersion} gave a second-order analysis of joint source-channel coding over finite alphabets based on constant composition codebooks. 
 
In this paper, we demonstrate that the nonasymptotic fundamental limit for coding over channels with cost constraints is closely approximated in terms of the cdf of a random variable we refer to as the $\mathsf b$-tilted information density, which parallels the notion of $\mathsf d$-tilted information for lossy compression \cite{kostina2011fixed}. 
We show a simple non-asymptotic converse bound for general channels with input cost constraints in terms of $\mathsf b$-tilted information density. Not only does this bound lead to a general strong converse result, but it is also tight enough to find the channel dispersion-cost function and the third order term equal to $\frac 1 2 \log n$ when coupled with the corresponding achievability bound.  
More specifically, we show that for the DMC, $\log M^\star(n, \epsilon, \beta)$, the logarithm of the maximum achievable code size at blocklength $n$, error probability $\epsilon$ and cost $\beta$, is given by, under mild regularity assumptions 
\begin{equation}
\log M^\star(n, \epsilon, \beta) = n C(\beta) - \sqrt{n V(\beta)} \Qinv{\epsilon} + \frac 1 2 \log n + \bigo{1} \label{eq:2ordercostintro}
\end{equation}
where $V(\beta)$ is the dispersion-cost function, 
thereby refining Hayashi's result \cite{hayashi2009secondorder} and providing a matching converse to the result of Moulin \cite{moulin2012logvolume}. We observe that the capacity-cost and the dispersion-cost functions are given by the mean and the variance of the $\mathsf b$-tilted information density. This novel interpretation juxtaposes nicely with the corresponding results in 
\cite{kostina2011fixed} ($\mathsf d$-tilted information in rate-distortion theory). 
Furthermore, we generalize \eqref{eq:2ordercostintro} to lossy joint source-channel coding of general memoryless sources over channels with cost. 

Section~\ref{sec:btilted} introduces the $\mathsf b$-tilted information density. Section~\ref{sec:bounds} states the new non-asymptotic converse bound which holds for a general channel with cost constraints, without making any assumptions on the channel (e.g. alphabets, stationarity, memorylessness). An asymptotic analysis of the converse and achievability bounds, including the proof of the strong converse and the expression for the channel dispersion-cost function, is presented in Section \ref{sec:asymptotic} in the context of memoryless channels. Section~\ref{sec:jscc} generalizes the results in Sections \ref{sec:bounds} and \ref{sec:asymptotic} to the lossy joint source-channel coding setup. 

\section{$\mathsf b$-tilted information density}
\label{sec:btilted}
In this section, we introduce the concept of $\mathsf b$-tilted information density and several relevant properties in a general single-shot approach. 

Fix the transition probability kernel $P_{Y| X} \colon \mathcal X \to \mathcal Y$ and the cost function $\mathsf b \colon \mathcal X \mapsto [0, \infty]$.  In the application of this single-shot approach in Section \ref{sec:asymptotic}, $\mathcal {X}$, $\mathcal {Y}$, $P_{Y|X}$ and $\mathsf b$ will become $\mathcal{A}^n$, $\mathcal{B}^n$, $P_{Y^n|X^n}$ and $\mathsf b_n$, respectively. 
Denote
\begin{align}
\mathbb C(\beta) &= \sup_{
\substack
{
P_{ X} \colon \\
\E{ \mathsf b( X)} \leq \beta
}
} 
I( X;  Y) \label{eq:C(b)},\\
\lambda^\star &= \mathbb C^\prime(\beta). \label{eq:lambdastar}
 \end{align}
 
Since $\mathbb C(\beta)$ is non-decreasing concave function of $\beta$ \cite[Theorem 6.11]{csiszar2011information}, $\lambda^\star \geq 0$. 
 For random variables $Y$ and $\bar Y$ defined on the same space, denote
\begin{equation}
\imath_{Y \| \bar Y} (y) = \log \frac{dP_{Y}}{dP_{\bar Y}}(y) \label{eq:idiv}.
\end{equation}

If $Y$ is distributed according to $P_{Y | X = x}$, we abbreviate the notation as
\begin{align}
 \imath_{ X;  \bar Y}( x;  y) &= \log \frac{dP_{Y|X = x}}{dP_{\bar Y}}(y) \label{eq:ibar}.
\end{align}
in lieu of $\imath_{Y|X = x \| \bar Y} (y) $. 
The information density  $\imath_{ X;  Y}( x;  y)$  between realizations of two random variables with joint distribution $P_X P_{Y|X}$ follows by particularizing \eqref{eq:ibar} to $\{P_{Y|X}, P_{Y}\}$, where $P_X \to P_{Y|X} \to P_Y$\footnote{We write $P_{X} \to P_{Y|X} \to P_{Y}$ to indicate that $P_Y$ is the marginal of $P_X P_{Y|X}$, i.e. $P_{Y}(y) = \int_{\mathcal X} dP_{Y|X}(y|x)dP_{X}(x)$.}. In general, however, the function in \eqref{eq:ibar} does not require $P_{\bar Y}$ to be induced by any input distribution.

Further, define the function
\begin{equation}
 \jmath_{ X;  \bar Y} ( x;  y, \beta)  = \imath_{ X;  \bar Y}( x;  y) - \lambda^\star\left( \mathsf b( x) - \beta\right) \label{eq:btilted}.
\end{equation}

The special case of \eqref{eq:btilted} with $P_{\bar Y} = P_{Y^\star}$,  where $P_{Y^\star}$ is the unique output distribution that achieves the supremum in \eqref{eq:C(b)} \cite{kemperman1974shannoncapacity}, defines $\mathsf b$-tilted information density:

\begin{defn}[$\mathsf b$-tilted information density]
The $\mathsf b$-tilted information density between $x \in \mathcal X$ and $y \in \mathcal Y$ is  $\jmath_{ X;  Y^\star}( x;  y, \beta)$. 
\label{defn:btilted}
\end{defn}

 Since $P_{Y^\star}$ is unique even if there are several (or none) input distributions $P_{X^\star}$ that achieve the supremum in \eqref{eq:C(b)}, there is no ambiguity in Definition \ref{defn:btilted}. If there are no cost constraints 
 (i.e. $\mathsf b(x) = 0 ~ \forall x \in \mathcal X$),  then $\mathbb C^\prime(\beta) = 0$ regardless of $\beta$, and 
\begin{equation}
 \jmath_{ X;  \bar Y}( x;  y, \beta) = \imath_{ X;  \bar Y}( x;  y).
\end{equation}
The counterpart of the $\mathsf b$-tilted information density in rate-distortion theory is the $\mathsf d$-tilted information \cite{kostina2011fixed}. 

\begin{example} For $n$ uses of a memoryless AWGN channel with unit noise power and maximal power not exceeding $nP$, $\mathbb C(P) = \frac n 2 \log (1 + P)$, and the output distribution that achieves \eqref{eq:C(b)} is $Y^{n \star} \sim \mathcal N\left(0, \left(1 + P\right) \mathbf I \right)$. Therefore
\begin{align}
\jmath_{X^n; Y^{n \star}}(x^n; y^n, P) &= \frac n 2 \log \left( 1 + P\right) - \frac {\log e} 2 \left| y^n - x^n \right|^2 
\notag\\
&
+ \frac {\log e} {2 (1 + P)}\left( \left|y^n\right|^2 - \left|x^n\right|^2 + n P \right), 
\end{align}
where the Euclidean norm is denoted by $\left|x^n\right|^2 = \sum_{i = 1}^n x_i^2$. 
It is easy to check that under $P_{Y^n| X^n = x^n}$, the distribution of $\jmath_{X^n; Y^{n \star}}(x^n; Y^n, P)$ is the same as that of (by `$\sim$' we mean equality in distribution)
\begin{align}
&~
\jmath_{X^n; Y^{n \star}}(x^n; Y^n, P) 
\notag \\
\sim
&~ 
\frac n 2 \log \left( 1 + P\right) 
- \frac {P \log e}{2 (1 + P)} 
\left[ W^n_{ \frac{\left|x^n\right|^2}{P^2}} - 
 n   - \frac{ \left| x^n\right|^2}{P^2}  \right],
 \label{eq:btiltedAWGN}
\end{align}
 where $W_\lambda^\ell$ denotes a non central chi-square distributed random variable with $\ell$ degrees of freedom and non-centrality parameter $\lambda$. The mean of \eqref{eq:btiltedAWGN} is $\frac n 2 \log \left( 1 + P\right)$, in accordance with \eqref{eq:C(b)jstarconditional}, while its variance is 
 $
 \frac 1 2 \frac{\left( n P^2 + 2 \left|x^n\right|^2 \right)}{(1 + P)^2} \log^2 e
 $ which becomes $n V(P)$ (found in \cite{polyanskiy2010channel} and displayed in \eqref{eq:Vawgn}) after averaging with respect to $X^{n \star}$ distributed according to $P_{X^{n \star}} \sim \mathcal N(0, P \mathbf I)$. 
\end{example}

Denote \footnote{We allow $\beta_{\max} = +\infty$.}
\begin{align}
 \beta_{\min} &= \inf_{x \in \mathcal X} \mathsf b(x), \\
 \beta_{\max} &= \sup \left\{ \beta \geq 0 \colon \mathbb C(\beta) < \mathbb C(\infty) \right\}.
\end{align}

Theorem \ref{thm:btilted} below highlights the importance of $\mathsf b$-tilted information density in the optimization problem \eqref{eq:C(b)}. Of key significance in the asymptotic analysis in Section \ref{sec:asymptotic}, Theorem \ref{thm:btilted} gives a nontrivial generalization of the well-known properties of information density to the setting with cost constraints. 
\begin{thm}
Fix $ \beta_{\min} < \beta < \beta_{\max}$. 
Assume that $P_{X^\star}$ achieving \eqref{eq:C(b)} is such that the constraint is achieved with equality:
\begin{equation}
 \E{ \mathsf b(X^\star)} = \beta.
\end{equation}

Then, the following equalities hold.
\begin{align}
\mathbb C(\beta) 
&= \sup_{P_X} \E{ \jmath_{ X;  Y}(X;  Y, \beta)} \label{eq:C(b)maxj}\\
&=  \sup_{P_X} \E{\jmath_{ X;  Y^\star}( X;  Y, \beta) } \label{eq:C(b)maxjstar}\\
&= \E{ \jmath_{ X;  Y^\star}( X^\star;  Y^\star, \beta) } \label{eq:C(b)jstar}\\
&= \E{ \jmath_{ X;  Y^\star}( X^\star;  Y^\star, \beta) | X^\star} \label{eq:C(b)jstarconditional},
\end{align}
where \eqref{eq:C(b)jstarconditional} holds $P_{X^\star}$-a.s., and $P_{ X} \to P_{ Y |  X} \to P_{ Y}$, $P_{ X^\star} \to P_{Y|X} \to P_{ Y^\star}$. 
\label{thm:btilted}
\end{thm}
\begin{proof}
Appendix  \ref{appx:btilted}. 
\end{proof}
Throughout the paper, we assume that the assumptions of Theorem \ref{thm:btilted} hold. 

For channels without cost, the inequality
\begin{equation}
 D(P_{Y|X = x} \| P_{Y^\star}) \leq C ~ \forall x \in \mathcal X
\end{equation}
is key to proving strong converses. Theorem \ref{thm:btilted} generalizes this result to channels with cost, showing that
\begin{equation}
\E{ \jmath_{X; Y^\star}(x; Y, \beta) | X = x} \leq C(\beta) ~ \forall x \in \mathcal X \label{eq:ineqcost}.
\end{equation}
\apxonly{
and, in particular
\begin{equation}
  D(P_{Y|X = x} \| P_{Y^\star}) \leq C(\beta) ~  \forall x \in \mathcal X \colon \mathsf b(x) \leq \beta \label{eq:quasicaod}
\end{equation}
}
 Note that \eqref{eq:ineqcost} is crucial for showing both the strong converse and the refined asymptotic analysis. 

\begin{remark}
 The general strong converse result in \cite[Theorem 3]{polyanskiy2011relativeentropy} includes channels with cost using the concept of `quasi-caod', which is defined as any output distribution $P_{Y^n}$ such that
\begin{equation}
 D(P_{Y^n|X^n = x^n} \| P_{Y^n}) \leq I^\star_n + \smallo{I^\star_n}  ~  \forall x^n \in \mathcal A^n \colon \mathsf b_n(x^n) \leq \beta,
\end{equation}
where $\mathcal A$ is the single-letter channel input alphabet, and
$I^\star_n = \max_{P_{X^n} \colon \mathsf b(X^n) \leq b \text{ a.s.}} I(X^n; Y^n)$.  
Since $C(\beta) = \lim_{n \to \infty} \frac 1 n I_n^\star$, \eqref{eq:ineqcost} implies that $P_{\mathsf Y^\star} \times \ldots \times P_{\mathsf Y^\star}$  is always a quasi-caod.  
\label{rem:strong}
\end{remark}

\begin{cor}
For all $P_X \ll P_{X^\star}$
\begin{align}
\Var{\jmath_{ X;  Y^\star}( X;  Y, \beta)} 
&=
\E{\Var{\jmath_{ X;  Y^\star}( X;  Y, \beta) |  X}} \label{eq:varbtilted}\\
&= 
\E{\Var{\imath_{ X;  Y^\star}( X;  Y) |  X}} \label{eq:varinfodensity}.
\end{align}
\label{cor:btiltedvar}
\end{cor}
\begin{proof}
Appendix  \ref{appx:btiltedvar}.
\end{proof}

\apxonly{
\begin{example}
Consider the memoryless binary symmetric channel (BSC) with crossover probability $\delta$ and Hamming per-symbol cost, $\mathsf b (\mathsf x) = \mathsf x$. The capacity-cost function is given by
\begin{equation}
C(\beta) = 
\begin{cases}
h(\beta \star \delta) -  h(\delta) & \beta \leq \frac 1 2\\
1 - h(\delta) & \beta > \frac 1 2 
\end{cases}
\end{equation}
 where $\beta \star \delta = (1 - \beta) \delta + \beta (1 - \delta)$. The capacity-cost function is achieved by 
 $P_{\mathsf X^\star}(1) = \min\left\{ \beta, \frac 1 2 \right\}$, and 
 $C^\prime(\beta) = \left( 1 - 2 \delta\right) \log \frac{1 - \beta \star \delta}{\beta \star \delta}$ for $\beta \leq \frac 1 2$, and
\begin{equation}
\jmath_{\mathsf X; \mathsf Y^\star}(\mathsf X^\star; \mathsf Y^\star, \beta) = 
h(\beta \star \delta) - h(\delta) + 
\begin{cases}
 \delta \log \frac{1 - \delta}{\delta} \frac {\beta \star \delta}{1 - \beta \star \delta} & \text{w.p.} ~ (1 - \delta)(1 - \beta)\\
  - (1-\delta) \log \frac{1 - \delta}{\delta} \frac {\beta \star \delta}{1 - \beta \star \delta}  & \text{w.p.} ~ \delta(1 - \beta)\\
  \delta \log \frac{1 - \delta}{\delta} \frac {1 - \beta \star \delta} {\beta \star \delta} & \text{w.p.} ~ (1 - \delta) \beta \\
- (1 - \delta) \log \frac{1 - \delta}{\delta} \frac {1 - \beta \star \delta} {\beta \star \delta} & \text{w.p.} ~   \delta \beta 
\end{cases} 
\end{equation}

The capacity-cost function (the mean of $\jmath_{\mathsf X; \mathsf Y^\star}(\mathsf X^\star; \mathsf Y^\star, \beta)$) and the dispersion-cost function (the variance of $\jmath_{\mathsf X; \mathsf Y^\star}(\mathsf X^\star; \mathsf Y^\star, \beta)$) are plotted in Fig. \ref{fig:binary}. 
\label{example:binary}

\begin{figure}[htbp]
\subfigure[]{
    \epsfig{file=Cbinary.eps, width=.5\linewidth}
        \label{subfig:Cbinary}
}
\subfigure[]{
    \epsfig{file=Vbinary.eps, width=.5\linewidth}
            \label{subfig:Vbinary}
}
\caption{ Capacity-cost function \subref{subfig:Cbinary} and the dispersion-cost function \subref{subfig:Vbinary} for BSC with crossover probability $\delta$ where the normalized Hamming weight of codewords is constrained not to exceed $\beta$. }
\label{fig:binary}
\end{figure}

\end{example}
}

\section{Nonasymptotic bounds}
\label{sec:bounds}

Converse and achievability bounds give necessary and sufficient conditions, respectively, on $(M, \epsilon, \beta)$ in order for a code to exist with $M$ codewords and average error probability not exceeding $\epsilon$ and cost not exceeding $\beta$. Such codes (allowing stochastic encoders and decoders) are rigorously defined next. 
 
\begin{defn}[$(M, \epsilon, \beta)$ code]
An $(M, \epsilon, \beta)$ code for $\{P_{Y|X}, \mathsf b\}$ is a pair of random transformations $P_{X|S}$ (encoder) and $P_{Z|Y}$ (decoder) such that $\Prob{ S \neq Z} \leq \epsilon$, where $S - X - Y - Z$, the probability is evaluated with $S$ equiprobable on an alphabet of cardinality $M$,  and the codewords satisfy the maximal cost constraint (a.s.)
\begin{equation}
\mathsf b(X) \leq \beta \label{eq:costmax}.
\end{equation}
\label{defn:(M, eps, alpha)}
\end{defn}

The non-asymptotic quantity of principal interest is $M^\star(\epsilon, \beta)$, the maximum code size achievable at error probability $\epsilon$ and cost $\beta$.

\begin{thm}[Converse]
 The existence of an $(M, \epsilon, \beta)$ code for $\{P_{Y|X}, \mathsf b\}$ requires that
\begin{align}
\epsilon \geq 
\max_{\gamma > 0} \bigg\{
&~\sup_{
 \bar Y
} 
  \inf_{x \colon \mathsf b(x) \leq \beta }  \Prob{\imath_{X; \bar Y}(x; Y) \leq \log M - \gamma | X = x} 
\notag\\&
  - \exp(-\gamma)   \bigg\}
  \label{eq:wolfowitz}\\
  \geq
  \max_{\gamma > 0} \bigg\{
&~\sup_{
 \bar Y 
}   
\inf_{x \in \mathcal X} 
\Prob{\jmath_{X; \bar Y}(x; Y, \beta) \leq \log M - \gamma| X = x}
\notag\\&
- \exp(-\gamma) \bigg\} \label{eq:Ccosta}.
\end{align}
\label{thm:Ccost}
\end{thm}

\begin{proof}
The bound in \eqref{eq:wolfowitz} is due to Wolfowitz \cite{wolfowitz1968notesongeneral}. The bound in \eqref{eq:Ccosta} simply weakens \eqref{eq:wolfowitz} using $\mathsf b(x) \leq \beta$. 
\end{proof}


By restricting the channel input space appropriately, converse bounds for channels with cost constraints can be obtained from the converse bounds in \cite{polyanskiy2010channel,kostina2012jscc}. Their analysis becomes tractable by the introduction of $\mathsf b$-tilted information density in \eqref{eq:Ccosta} and  an application of \eqref{eq:ineqcost}.

Achievability bounds for channels with cost constraints can be obtained from the random coding bounds in \cite{polyanskiy2010channel,kostina2012jscc} by restricting the distribution from which the codewords are drawn to satisfy $\mathsf b(X) \leq \beta$ a.s.  
In particular, for the DMC, we may choose $P_{X^n}$ to be equiprobable on the set of codewords of the type closest (among types satisfying the cost constraint) to the input distribution $P_{\mathsf X^\star}$ that achieves the capacity-cost function. As shown in \cite{hayashi2009secondorder}, 
such constant composition codes achieve the dispersion of channel coding under input cost constraints.
Unfortunately, the computation of such bounds may become challenging in high dimension, particularly with continuous alphabets. 

\apxonly{
The following result is a cost-constrained version of Shannon's achievability bound \cite{shannon1957certain}. 
\begin{thm}[Achievability]
 There exists an $(M, \epsilon, \beta)$ code for $\{P_{Y|X}, \mathsf b\}$ such that
\begin{align}
\epsilon 
&\leq
\inf_{\gamma > 0} \left\{ \inf_{X} \left\{\Prob{\imath_{X; Y}(X; Y) \leq \log(M - 1) + \gamma} + \Prob{\mathsf b(X) > \beta} \right\}+ \exp(-\gamma)\right\}  \label{eq:A}\\
\end{align}
\label{thm:A}
\end{thm}

The following result is a cost-constrained version of Dependence Testing bound  \cite{polyanskiy2010channel}. 
\begin{thm}[Dependence Testing bound \cite{polyanskiy2010channel}]
There exists an $(M, \epsilon, \beta)$ code with 
\begin{equation}
\epsilon \leq \inf_{P_X} \left\{ \E{\exp\left( - \left| \imath_{X; Y}(X; Y) - \log \frac{M - 1}{2}\right|^+\right) } + \Prob{\mathsf b(X) > \beta} \right\} \label{eq:DT} 
\end{equation}

\begin{equation}
\epsilon \leq \inf_{P_X} \left\{ \E{\exp\left( - \left| \imath_{X; Y}(X; Y) - \lambda \mathsf b(X) + \lambda \beta - \lambda \delta - \log \frac{M - 1}{2}\right|^+\right) } + \Prob{\mathsf b(X) > \beta} + \Prob{\mathsf b(X) < \beta - \delta }  \right\} \label{eq:DT} 
\end{equation}

\end{thm}
}

\apxonly{
The following is a trivial weakening of the standard achievability result by Feinstein \cite{feinstein1954new}, which, when coupled with the converse result in Theorem \ref{thm:Ccost}, 
enables to find the channel dispersion-cost function. 
\begin{thm}[Achievability]
 There exists an $(M, \epsilon, \beta)$ code for $P_{Y|X}$ such that
\begin{align}
\epsilon 
&\leq
\inf_{\gamma > 0} \left\{ \inf_{P_X \colon \mathsf b(X) \leq \beta} \Prob{\imath_{X; Y}(X; Y) \leq \log(M - 1) + \gamma} + \exp(-\gamma)\right\}  \label{eq:A}\\
&\leq
\inf_{\gamma, \delta > 0} \left\{ \inf_{P_X \colon \beta - \delta \leq \mathsf b(X) \leq \beta} \Prob{\jmath_{X; Y}(X; Y, \beta) \leq \log(M - 1) + \gamma + \lambda^\star \delta} + \exp(-\gamma)\right\}  \label{eq:Adelta}
\end{align}
\label{thm:A}
\end{thm}
}

\section{Asymptotic analysis}
\label{sec:asymptotic}
To introduce the blocklength into the non-asymptotic converse of Section \ref{sec:bounds}, we consider $(M, \epsilon, \beta)$ codes for $\{P_{Y^n|X^n}, \mathsf b_n\}$, where $P_{Y^n|X^n}\colon \mathcal A^n \mapsto \mathcal B^n$ and $\mathsf b_n \colon \mathcal A^n \mapsto [0, \infty]$. We call such codes $(n, M, \epsilon, \beta)$ codes, and denote the corresponding non-asymptotically achievable maximum code size by $M^\star(n, \epsilon, \beta)$.

\subsection{Assumptions}
\label{sec:assumptions}
The following basic assumptions hold throughout Section \ref{sec:asymptotic}. 
\begin{enumerate}[(i)]
\item The channel is stationary and memoryless, $P_{Y^n | X^n}  = P_{\mathsf Y| \mathsf X} \times \ldots \times P_{\mathsf Y | \mathsf X}$.
 \label{item:ch}
 \item The cost function is separable, 
 $ \mathsf b_n(x^n) = \frac 1 n \sum_{i = 1}^n \mathsf b(x_i)
 $,
 where $\mathsf b \colon \mathcal A \mapsto [0, \infty]$. 
  \label{item:separablecost}
 \item Each codeword is constrained to satisfy the maximal cost constraint, $\mathsf b_n(x^n) \leq \beta$.
\label{item:maximalcost}
\item 
$\sup_
 {\mathsf x \in \mathcal A} \Var{\jmath_{\mathsf X; \mathsf Y^\star}(\mathsf x; \mathsf Y, \beta)| \mathsf X = x} = V_{\max} < \infty$.
 \label{item:last} 
\end{enumerate}
Under these assumptions, the capacity-cost function is given by
\begin{equation}
C (\beta ) = \sup_{P_{\mathsf X} \colon \E{ \mathsf b(\mathsf X) } \leq \beta } I(\mathsf{X};\mathsf{Y})  \label{eq:capacity-cost}.
\end{equation}

Observe that in view of assumptions \eqref{item:ch} and \eqref{item:separablecost}, as long as $P_{\bar Y^n}$ is a product distribution, 
$P_{\bar Y^n} = P_{\bar{\mathsf Y}} \times \ldots \times P_{\bar{\mathsf Y}}$,
\begin{equation}
\jmath_{X^n; \bar Y^n}(x^n; y^n, \beta) = \sum_{i = 1}^n \jmath_{\mathsf X; \bar{\mathsf Y}}(x_i; y_i, \beta).
\end{equation}

\subsection{Strong converse}

Although the tools developed in Sections \ref{sec:btilted} and \ref{sec:bounds} are able to result in a strong converse for channels that exhibit ergodic behavior (see also Remark \ref{rem:strong}), for the sake of concreteness and length, we only deal here with the memoryless setup described in Section \ref{sec:assumptions}. 

We show that if transmission occurs at a rate greater than the capacity-cost function,  the error probability must converge to $1$,  regardless of the specifics of the code.  Towards this end, we fix some $\alpha > 0$, we choose $\log M \geq n C(\beta) + 2 n \alpha$, and we weaken the bound \eqref{eq:Ccosta} in Theorem \ref{thm:Ccost} by fixing $\gamma = n \alpha$ and $P_{\bar Y^n} = P_{\mathsf Y^\star} \times \ldots \times P_{\mathsf Y^\star}$, where $\mathsf Y^\star$ is the output distribution that achieves $C(\beta)$,  to obtain
\begin{align}
 \epsilon &\geq \inf_{x^n \in \mathcal A^n} 
\Prob{\sum_{i = 1}^n \jmath_{\mathsf X; \mathsf Y^\star}(x_i; Y_i, \beta) \leq n C(\beta) + n \alpha }
\notag\\&
-  \exp(- n \alpha)  \\
&\geq \inf_{x^n \in \mathcal A^n} 
\Prob{\sum_{i = 1}^n \jmath_{\mathsf X; \mathsf Y^\star}(x_i; Y_i, \beta) \leq \sum_{i = 1}^n c(x_i)+ n \alpha }
\notag\\&
-  \exp(- n \alpha) \label{eq:-Cstrong},
\end{align}
 where for notational convenience we have abbreviated
\begin{equation}
 c( \mathsf x) = \E{ \jmath_{\mathsf X; \mathsf Y^\star}(\mathsf x; \mathsf Y, \beta)| \mathsf X = \mathsf x},
\end{equation}
and 
 \eqref{eq:-Cstrong} employs \eqref{eq:C(b)maxjstar}.

To show that the right side of \eqref{eq:-Cstrong} converges to $1$, we invoke the following law of large numbers for non-identically distributed random variables.

\begin{lemma}[e.g. \cite{cinlar2011probability}]
 Suppose that $W_i$ are uncorrelated and $\sum_{i = 1}^\infty \Var{\frac{W_i}{c_i}} < \infty$ for some strictly positive sequence $(c_n)$ increasing to $+\infty$. Then, 
\begin{equation}
\frac 1 {c_n} \left(\sum_{i = 1}^n W_i - \E{\sum_{i = 1}^n W_i} \right) \to 0 \text{ in } L^2.
\end{equation}
\label{lemma:LLN}
\end{lemma}
Let $W_i = \jmath_{\mathsf X; \mathsf Y^\star}(x_i; Y_i, \beta)$ and $c_i = i$. 
Since (recall \eqref{item:last})
\begin{align}
 \sum_{i = 1}^\infty \Var{\frac 1 i \jmath_{\mathsf X; \mathsf Y^\star}(x_i; Y_i, \beta) | X_i = x_i} 
 &\leq V_{\max} \sum_{i = 1}^\infty \frac 1 {i^2} 
 \\
 &< \infty,
\end{align}
by virtue of Lemma \ref{lemma:LLN}, the right side of \eqref{eq:-Cstrong} converges to $1$, so any channel satisfying \eqref{item:ch}--\eqref{item:last} also satisfies the strong converse.  

As noted in \cite[Theorem 77]{polyanskiy2010thesis} in the context of the AWGN channel,
the strong converse does not hold if the cost constraint is averaged over the codebook, i.e. if, in lieu of \eqref{eq:costmax}, the cost requirement is  
\begin{equation}
\frac 1 M \sum_{m= 1}^M \E{\mathsf b(X) | S = m }\leq \beta \label{eq:costav}.
\end{equation}
To see why the strong converse does not hold in general, fix a code of rate $C(\beta) < R < C(2\beta)$ none of whose codewords cost more than $2\beta$ and whose error probability satisfies $\epsilon_n \to 0$. Since $R < C(2\beta)$, such a code exists. Now, replace half of the codewords with the all-zero codeword (assuming $\mathsf b(\mathbf 0) = 0$) while leaving the decision regions of the remaining codewords untouched. The average cost of the new code satisfies \eqref{eq:costav}, its rate is greater than the capacity-cost function, $R > C(\beta)$, yet its average error probability does not exceed $\epsilon_n + \frac 1 2 \to \frac 1 2$.  
\subsection{Dispersion}
\label{sec:2orderCost}
First, we give the operational definition of the dispersion-cost function of any channel. 
\begin{defn}[Dispersion-cost function]
The channel dispersion-cost function, measured in squared information units per channel use, is defined by
\begin{equation}
 V(\beta) = \lim_{\epsilon \to 0} \limsup_{n \to \infty} \frac 1 n \frac{\left( nC(\beta) - \log M^\star(n, \epsilon, \beta)\right)^2}{2 \log_e \frac 1 \epsilon} \label{eq:dispersioncostdefn}.
\end{equation}
\end{defn}
An explicit expression for the dispersion-cost function of a discrete memoryless channel is given in the next result.

\begin{thm}
In addition to assumptions \eqref{item:ch}--\eqref{item:last}, assume that the capacity-achieving input distribution $P_{\mathsf X^\star}$ is unique and that the channel has finite input and output alphabets. 
\begin{align}
\log M^\star(n, \epsilon, \beta) &= n C(\beta) - \sqrt{n V(\beta)} \Qinv{\epsilon} + \theta(n) \label{eq:2ordercost}, \\
C(\beta) &= \E{\jmath_{ \mathsf X;  \mathsf Y^\star}( \mathsf X^\star;  \mathsf Y^\star, \beta)}, \\
V(\beta) &= \Var{\jmath_{ \mathsf X;  \mathsf Y^\star}( \mathsf X^\star;  \mathsf Y^\star, \beta)},  \label{eq:dispersioncost}
\end{align}

where 
the remainder term $\theta(n)$ satisfies:

\begin{enumerate}[a)]

\item If $V(\beta) > 0$,
\begin{align}
- \frac 1 2 \left( 
 \left| \mathrm{supp}\left(P_{{ \mathsf X}^\star}\right) \right| - 1\right) \log n + \bigo{1} &\leq \theta(n) \label{eq:remainderCostA}\\
&\leq \frac 1 2 \log n + \bigo{1}.\label{eq:remainderCostC}
\end{align}

\item If $V(\beta) = 0$, \eqref{eq:remainderCostA} holds, and \eqref{eq:remainderCostC} is replaced by 
\end{enumerate}
\begin{align}
\theta(n) \leq \bigo{n^{\frac 1 3}} \label{eq:remainderCost0}.
\end{align}

\label{thm:2orderCost}
\end{thm}

\begin{proof}[Proof]

{\it Converse.} Full details are given in Appendix \ref{appx:2orderCost}. The main steps of the refined asymptotic analysis of the bound in Theorem \ref{thm:Ccost} are as follows. First, building on the ideas of \cite{polyanskiy2013saddle,tomamichel2013thirdorder}, we weaken the bound in \eqref{eq:Ccosta} by a  careful choice of a non-product auxiliary distribution $P_{\bar Y^n}$.  Second, using Theorem \ref{thm:btilted} and the technical tools developed in Appendix \ref{appx:auxiliary}, we show that the infimum in the right side of \eqref{eq:Ccosta} is lower bounded by $\epsilon$ for the choice of $M$ in \eqref{eq:2ordercost}.

{\it Achievability.} Full details are given in Appendix \ref{appx:2orderCostA}, which provides an asymptotic analysis of the Dependence Testing bound of \cite{polyanskiy2010channel} in which the random codewords are of type closest to $P_{\mathsf X^\star}$, rather than drawn from the product distribution $P_{\mathsf X} \times \ldots \times P_{\mathsf X}$, as in achievability proofs for channel coding without cost constraints. We use Corollary \ref{cor:btiltedvar} to establish that such constant composition codes achieve the dispersion-cost function. 
\end{proof}

\begin{remark}
According to a recent result of Moulin \cite{moulin2012logvolume}, the achievability bound on the remainder term in \eqref{eq:remainderCostA} can be tightened to match the converse bound in \eqref{eq:remainderCostC}, thereby establishing that 
\begin{equation}
 \theta(n) = \frac 1 2 \log n + \bigo{1}, \label{eq:moulin}
\end{equation}
 provided that the following regularity assumptions hold:
\begin{itemize}
\item The random variable $\imath_{\mathsf X; \mathsf Y^\star}(\mathsf X^\star; \mathsf Y^\star)$ is of nonlattice type; 
\item $\mathrm{supp}(P_{\mathsf X^\star}) = \mathcal A$;
\item 
$\mathrm{Cov}
\left[ \imath_{\mathsf X; \mathsf Y^\star}(\mathsf X^\star; \mathsf Y^\star), 
\imath_{\mathsf X; \mathsf Y^\star}(\bar {\mathsf X}^\star; \mathsf Y^\star) 
\right] 
< \Var{\imath_{\mathsf X; \mathsf Y^\star}(\mathsf X^\star; \mathsf Y^\star)}
$ where
 
$
P_{ {\bar {\mathsf X}^\star} \mathsf X^\star \mathsf Y^\star}(\bar {\mathsf x}, \mathsf x, \mathsf y) = \frac{1}{P_{\mathsf Y^\star}(\mathsf y)} P_{\mathsf X^\star}(\bar{\mathsf x}) P_{\mathsf Y | \mathsf X}(\mathsf y | \bar {\mathsf x}) P_{\mathsf Y | \mathsf X}(\mathsf y | \mathsf x) P_{\mathsf X^\star}(\mathsf x)
$.
\end{itemize}
\end{remark}

\begin{remark}
As we show  in Appendix \ref{appx:continuous}, Theorem \ref{thm:2orderCost} applies to channels with abstract alphabets provided that in addition to \eqref{item:ch}--\eqref{item:separablecost}, they meet the following criteria:

\begin{enumerate}[(a)]
 \item The cost function $\mathsf b \colon \mathcal A \to [0, \infty]$ is such that for all $\gamma \in [\beta, \infty)$,  $\mathsf b^{-1}(\gamma)$ is nonempty. In particular, this condition is satisfied if the channel input alphabet $\mathcal A$ is a metric space, and $\mathsf b$ is continuous and unbounded with $\mathsf b(0) = 0$.
 \label{item:costunbounded}
 
 \item The distribution of $\imath_{X^n; Y^{n \star}}(x^n; Y^n)$, where $P_{Y^{n \star}} = P_{\mathsf Y^\star} \times \ldots \times P_{\mathsf Y^\star}$ does not depend on the choice of $x^n \in \mathcal F_n$, where $\mathcal F_n = \{x^n \in \mathcal A^n \colon \mathsf b_n(x^n) = \beta\}$.
  \label{item:symmetry}
  
 \item For all $\mathsf x$ in the projection of $\mathcal F_n$ onto $\mathcal A$, i.e. for all $\mathsf x$ such that $(\mathsf x, x_2, \ldots, x_n) \in \mathcal F_n$ for some $x_2, \ldots, x_n$,
\begin{equation}
 \E{\left| \jmath_{\mathsf X; {\mathsf Y}^\star}(\mathsf X; \mathsf Y, \beta) - C(\beta) \right|^3 | \mathsf X = \mathsf x }  < \infty.
\end{equation}
\label{item:thirdmoment}

 \item\hspace{-1.5mm}\footnote{For the converse result, assumptions \eqref{item:costunbounded}--\eqref{item:thirdmoment} suffice. }
 There exists a distribution $P_{X^n}$ supported on $\mathcal F_n$ such that $\imath_{Y^n \| Y^{n \star}}(Y^n)$, where $P_{X^n} \to P_{Y^n | X^n} \to P_{Y^n}$, is almost surely bounded by $f_n = \smallo{\sqrt n}$ from above. 
\label{item:idivbounded}

\end{enumerate}

Then, \eqref{eq:2ordercost} holds identifying \eqref{eq:capacitycostsym}--\eqref{eq:remainderGauss} or all $\mathsf x \in \mathcal A$ s.t. $\mathsf b(\mathsf x) = \beta$:
\begin{align}
C(\beta) &= D(P_{\mathsf Y | \mathsf X = \mathsf x} \| P_{\mathsf Y^\star}) \label{eq:capacitycostsym},\\
V(\beta) &= \Var{\imath_{\mathsf X; \mathsf Y^\star}(\mathsf x; \mathsf Y) | \mathsf X = \mathsf x} \label{eq:dispersioncostsym},\\
- f_n + \bigo{1}
&\leq \theta(n)
\leq \frac 1 2 \log n + \bigo{1}, \label{eq:remainderGauss}
\end{align}
where $f_n  =\smallo{\sqrt n}$ is specified in \eqref{item:idivbounded}.

\label{remark:continuous}
\end{remark}

\begin{remark}
 Theorem~\ref{thm:2orderCost} with the remainder in \eqref{eq:moulin} \cite{tan2014thirdorder} also holds for the AWGN channel with maximal signal-to-noise ratio $P$, offering a novel interpretation of the dispersion of the Gaussian channel \cite{polyanskiy2010channel}
\begin{equation}
 V(P) = \frac 1 2 \left( 1 - \frac 1 {\left( 1 + P \right)^2}\right) \log^2 e \label{eq:Vawgn}
\end{equation}
as the variance of the $\mathsf b$-tilted information density. We note that the AWGN channel satisfies the conditions of Remark \ref{remark:continuous} with $P_{X^n}$ uniform on the power sphere and $f_n = \bigo{1}$ \cite{polyanskiy2010channel}.
\end{remark}

\begin{remark}
As we show in Appendix \ref{appx:exp}, a stationary memoryless channel with $\mathsf b( \mathsf x) = \mathsf x$ which takes a nonnegative input and adds an exponential noise of unit mean to it \cite{verdu1996exponential}, satisfies the conditions of Remark \ref{remark:continuous} with $f_n = \bigo{1}$, and
\begin{align}
\jmath_{\mathsf X; \mathsf Y^\star} (\mathsf x; \mathsf y, \beta) &= \log (1 + \beta)  + \frac {\beta} {1 + \beta} \left(\mathsf x- \mathsf y + 1\right)  \log e, \label{eq:btiltedexp} \\
C(\beta) &= \log (1 + \beta), \label{eq:Cexp} \\
V(\beta) &= \frac {\beta^2} {(1 + \beta)^2} \log^2 e. \label{eq:Vexp}
\end{align}
\end{remark}

\begin{remark}
As should be clear from the proof of Theorem \ref{thm:2orderCost}, if the capacity-achieving distribution is not unique, then
\begin{equation}
 V(\beta) =
\begin{cases}
 \min \Var{\jmath_{ \mathsf X;  \mathsf Y^\star}( \mathsf X^\star;  \mathsf Y^\star, \beta)} & 0 < \epsilon \leq \frac 1 2\\
  \max \Var{\jmath_{ \mathsf X;  \mathsf Y^\star}( \mathsf X^\star;  \mathsf Y^\star, \beta)} & \frac 1 2 < \epsilon < 1
\end{cases}
\end{equation}
where the optimization is performed over all $P_{\mathsf X^\star}$ that achieve $C(\beta)$. This parallels the dispersion result for channels without cost \cite{polyanskiy2010channel}. 
\end{remark}

\apxonly{
The converse bound in Theorem \ref{thm:Ccost}, the matching achievability bound in \cite[Theorem 17]{polyanskiy2010channel}, and the Gaussian approximation in Theorem \ref{thm:2orderCost} in which the remainder is approximated by $\theta(n) \approx \frac 1 2 \log n$ are plotted in Figures \ref{fig:binary025} and \ref{fig:binary04} for the BSC with Hamming cost discussed in Example \ref{example:binary}. 
As evidenced by the plots, although the minimum over the channel inputs in \eqref{eq:Ccosta} may be difficult to analyze, it is not difficult to compute (in polynomial time), at least for the DMC.

\begin{figure}[htbp]
    \epsfig{file=b_025.eps, width=1\linewidth}
        \label{subfig:Cbinary}
\caption{Maximum achievable rate for BSC with crossover probability $\delta = 0.11$ where the normalized Hamming weight of codewords is constrained not to exceed $\beta = 0.25$ and the tolerated error probability is $\epsilon = 10^{-4}$. }
\label{fig:binary025}
\end{figure}

\begin{figure}[htbp]
    \epsfig{file=b_04.eps, width=1\linewidth}
        \label{subfig:Cbinary}
\caption{Maximum achievable rate for BSC with crossover probability $\delta = 0.11$ where the normalized Hamming weight of codewords is constrained not to exceed $\beta = 0.4$ and the tolerated error probability is $\epsilon = 10^{-4}$. }
\label{fig:binary04}
\end{figure}
}

\section{Joint source-channel coding}
\label{sec:jscc}
In this section we state the counterparts of Theorems \ref{thm:Ccost} and \ref{thm:2orderCost} in the lossy joint source-channel coding setting. Proofs of the results in this section are obtained by fusing the proofs in 
Sections \ref{sec:bounds} and \ref{sec:asymptotic}
and those in \cite{kostina2012jscc}. 

In the joint source-channel coding setup the source is no longer equiprobable on an alphabet of cardinality $M$, as in Definition \ref{defn:(M, eps, alpha)},  rather it is arbitrarily distributed on an abstract alphabet $\mathcal M$.  Further, instead of reproducing the transmitted $S$ under a probability of error criterion, we might be interested in approximating $S$ within a certain distortion, so that a decoding failure occurs if the distortion between the source and its reproduction exceeds a given distortion level $d$, i.e. if $ \mathsf d(S, Z) > d$, where $Z \in \widehat{\mathcal M}$ is the representation of $S$, $\widehat {\mathcal M}$ is a reproduction alphabet, and $\mathsf d \colon \mathcal M \times \widehat{\mathcal M} \mapsto \mathbb R_+$ is the distortion measure. A $(d, \epsilon, \beta)$ code is a code for a fixed source-channel pair such that the probability of exceeding distortion $d$ is no larger than $\epsilon$ and no channel codeword costs more than $\beta$.  A $(d, \epsilon, \beta)$ code in a block coding setting, when a source block of length $k$ is mapped to a channel block of length $n$, is called a $(k, n, d, \epsilon, \beta)$ code.  The counterpart of the $\mathsf b$-tilted information density in lossy compression is the $\mathsf d$-tilted information, $\jmath_{S}(s, d)$, which can be computed using the equality
\begin{equation}
 \jmath_{S}(s, d) = \imath_{Z^\star;\, S} (z; s) + \lambda_S \mathsf d(s, z) -  \lambda_S d, \label{eq:dtilted}
\end{equation}
where $Z^\star$ is the random variable that achieves the infimum on the right side of
\begin{equation}
\mathbb R_S(d) \triangleq \min_{\substack{P_{Z|S} \colon \\ \E{\mathsf d(S, Z)} \leq d}} I(S; Z) ,
\end{equation}
$\lambda_S  = - \mathbb R_S^\prime(d) > 0$, and equality in \eqref{eq:dtilted} holds for $P_{Z^\star}$-a.e. $z$ \cite{kostina2011fixed}.  In a certain sense, the $\mathsf d$-tilted information quantifies the number of bits required to reproduce the source outcome $s \in \mathcal M$ within distortion $d$. For rigorous definitions and further details we refer the reader to \cite{kostina2012jscc}. 

\begin{thm}[Converse]
\label{thm:Ccostjscc}
The existence of a $(d, \epsilon, \beta)$ code for $S$ and $P_{Y|X}$ requires that
 \begin{align}
 \epsilon \geq&~  \inf_{P_{X|S}} \max_{\gamma > 0} \bigg\{ 
\sup_{
 \bar Y} 
\Prob{ \jmath_S(S, d) - \jmath_{X; \bar Y }(X; Y, \beta ) \geq \gamma }
  \notag\\
   &~
   - \exp\left( -\gamma \right) \bigg\}\label{eq:Ccostajscc}\\
   \geq &~\max_{\gamma > 0}\bigg\{ 
\sup_{
 \bar Y
}    
  \E{\inf_{x \in \mathcal X}  \Prob{ \jmath_S(S, d) -  \jmath_{X; \bar Y }(x; Y, \beta ) \geq \gamma \mid S} } 
   \notag\\ &~
 - \exp\left( -\gamma \right) \bigg\}\label{eq:Ccostjscc},
\end{align}
where the probabilities in \eqref{eq:Ccostajscc}  and \eqref{eq:Ccostjscc} are with respect to $P_S P_{X|S} P_{Y|X}$ and $P_{Y|X=x}$, respectively.
\end{thm}
\begin{proof}
 The bound is obtained by weakening \cite[Theorem 1]{kostina2012jscc} \eqref{eq:wolfowitz} using $\mathsf b(x) \leq \beta$. 
\end{proof}

Under the usual memorylessness assumptions, applying Theorem \ref{lemma:LLN} to the bound in \eqref{eq:Ccostjscc}, it is easy to show that the strong converse holds for lossy joint source-channel coding over channels with input cost constraints. A more refined analysis leads to the following result. 

\begin{thm}[Gaussian approximation]
\label{thm:2orderCostJSCC}
Assume the channel has finite input and output alphabets. For stationary memoryless sources satisfying the regularity assumptions (i)--(iv) of \cite{kostina2012jscc} and channels satisfying assumptions \eqref{item:separablecost}--\eqref{item:last} of Section \ref{sec:assumptions}, the parameters of the optimal $(k, n, d, \epsilon)$ code satisfy
\begin{equation}
 nC(\beta) - k R(d) = \sqrt{n V(\beta) + k \mathcal V(d)}\, \Qinv{ \epsilon} + \theta\left(n\right) \label{eq:2orderCostJSCC},
 \end{equation}
 where $\mathcal V(d) = {\rm Var}\left[\jmath_{\mathsf S}(\mathsf S, d) \right]$,  $V(\beta)$ is given in \eqref{eq:dispersioncost}, and the remainder $\theta\left(n\right)$ satisfies, if $V(\beta) > 0$,
\begin{align}
 &- \frac 1 2 \log n + \bigo{\sqrt{\log n}} 
 \leq  \theta(n) \\
 &\leq \bar \theta(n) + \left( \frac 1 2 | \mathrm{supp}(P_{\mathsf X^\star})| - 1\right) \log n,
\end{align}
where $\bar \theta(n) = \bigo{\log n}$ denotes the upper bound on the remainder term given in \cite[Theorem 10]{kostina2012jscc}.  If $V(\beta) = \mathcal V(d) = 0$, the upper bound on $\theta(n)$ stays the same, and the lower one becomes $\bigo{n^\frac 1 3}$. 
\end{thm}

\begin{proof}[Proof outline]
The achievability part is proven joining the asymptotic analyses of \cite[Theorem 8]{kostina2012jscc} and of Theorem \ref{thm:DT}, shown in Appendix \ref{appx:2orderCostA}. For the converse part, $P_{\bar Y}$ is chosen as in \eqref{eq:PYbar}, and similar to the proof of the converse part of \cite[Theorem 10]{kostina2012jscc},  a typical set of source outcomes is identified, and it is shown using Theorem \ref{thm:minprob}.\ref{thm:minprob:2} that for every source outcome in that set, the inner infimum in \eqref{eq:Ccostjscc} is approximately achieved by the capacity-achieving channel input type. 
\end{proof}

\section{Conclusion}
\label{sec:conclusion}
We introduced the concept of $\mathsf b$-tilted information density (Definition \ref{defn:btilted}), a random variable whose distribution governs the analysis of optimal channel coding under input cost constraints. The properties of $\mathsf b$-tilted information density listed in Theorem \ref{thm:btilted} play a key role in the asymptotic analysis of the converse bound in Theorem \ref{thm:Ccost} in Section \ref{sec:asymptotic}, which does not only lead to the strong converse and the dispersion-cost function when coupled with the corresponding achievability bound, but it also proves that the third order term in the asymptotic expansion \eqref{eq:2ordercostintro} is upper bounded (in the most common case of $V(\beta) > 0$) by $\frac 1 2 \log n + \bigo{1}$. In addition, we showed in Section \ref{sec:jscc} that the results of \cite{kostina2012jscc} generalize to coding over channels with cost constraints and also tightened the estimate of the third order term in \cite{kostina2012jscc}.  As propounded in \cite{polyanskiy2013saddle,tomamichel2013thirdorder}, the gateway to the refined analysis of the third order term is an apt choice of a non-product distribution $P_{\bar Y^n}$ in the bounds in Theorems \ref{thm:Ccost} and \ref{thm:Ccostjscc}.

\section{Acknowledgement}
We thank the referees for their unusually thorough reviews, which are reflected in the final version.

\appendices
\section{Proof of Theorem \ref{thm:btilted}}
\label{appx:btilted}

We note first two auxiliary results. 
\begin{lemma}[ \hspace{-1mm}\cite{csiszar1975Idivergence}]
 Let $0 \leq \alpha \leq 1$, and let $P \ll Q$ be distributions on the same probability space. Then,
\begin{equation}
 \lim_{\alpha \to 0} \frac 1 {\alpha} D(\alpha P+ (1-\alpha)Q \| Q) = 0.
 \label{eq:verdupertube}
\end{equation}
\label{lemma:verdupertube}
\end{lemma}

\begin{lemma}[Donsker-Varadhan \cite{donsker1975asymptotic}]
Let $g \colon \mathcal X \mapsto [-\infty, + \infty]$ and let $\bar X$ be a random variable on $\mathcal X$ such that 
$
\E{\exp\left( g(\bar X)\right) } < \infty
$. 
Then, 
\begin{equation}
 \E{ g(X) } - D(X \| \bar X) \leq \log \E{\exp\left( g(\bar X)\right) } \label{eq:lemmaverdu}
\end{equation}
with equality if and only if $X$ has distribution $P_{X^\star}$ such that
\begin{equation}
 \imath_{X^\star \| \bar X}(x) = g(x)  - \log \E{\exp\left( g(\bar X)\right) } \label{eq:lemmaverduXstar}.
\end{equation}
\label{lemma:verdu}
\end{lemma}
\begin{proof}
 If the left side of \eqref{eq:lemmaverdu} is not $- \infty$, we can write
\begin{align}
  \E{ g(X) } - D(X \| \bar X) &= \E{ g(X) - \imath_{X\| X^\star}(X) - \imath_{X^\star \| \bar X } (X)}\\
  &= \log \E{\exp\left( g(\bar X)\right) } - D(X \| X^\star),
\end{align}
which is maximized by letting $P_X = P_{X^\star}$. 
\end{proof}

We proceed to prove Theorem \ref{thm:btilted} by generalizing \cite[Theorem 6.1]{verduIT}. 
Equality \eqref{eq:C(b)maxj} is a standard result in convex optimization. By the assumption, the supremum in the right side of \eqref{eq:C(b)maxj} is attained by $P_{X^\star}$, therefore $\mathbb C(\alpha)$ is equal to the right side of \eqref{eq:C(b)jstar}. 

To show \eqref{eq:C(b)maxjstar}, fix $0 \leq \alpha \leq 1$. Denote 
\begin{align}
P_{\bar X} &\to P_{Y|X } \to P_{\bar Y},\\
P_{\hat X} &= \alpha P_{\bar X} + (1-\alpha)P_{X^\star} , \label {eq:PhatX} \\
P_{\hat X}  &\to P_{Y|X} \to P_{\hat Y} = \alpha P_{\bar Y} + (1-\alpha)P_{Y^\star},
\end{align}
and write
\begin{align}
&~\alpha
\left( 
\E{\jmath_{X; Y^\star} (X^\star; Y^\star, \beta)} 
- \E{\jmath_{X; Y^\star} (\bar X; \bar Y, \beta)} 
\right) 
\notag\\ &~
+ 
D(\hat Y \| Y^\star) \notag\\
=&~ \alpha D (P_{Y|X} \|P_{Y^\star} | P_{X^\star}) - \alpha D (P_{Y|X} \|P_{Y^\star} | P_{ \bar X}) + D(\hat Y \| Y^\star)
\notag \\ &~ 
+
\lambda^\star \alpha \E{\mathsf b(\bar X)} - \lambda^\star \alpha \E{\mathsf b(X^\star)}
\label{eq:-btilted1}\\
=&~  D (P_{Y|X} \|P_{Y^\star} | P_{X^\star}) - D (P_{Y|X} \|P_{Y^\star} | P_{ \hat X}) + D(\hat Y \| Y^\star)
\notag \\ &~ 
- \lambda^\star \E{\mathsf b(X^\star)} + \lambda^\star \E{\mathsf b(\hat X)}
\\
=&~ D (P_{Y|X} \|P_{Y^\star} | P_{X^\star}) - D(P_{Y|X} \|P_{\hat Y} | P_{\hat X}) 
- \lambda^\star \E{\mathsf b(X^\star)} 
\notag\\ &~
+ \lambda^\star \E{\mathsf b(\hat X)}\\
&=\E{\jmath_{X; Y^\star} (X^\star; Y^\star, \beta)}  - \E{\jmath_{X; \hat Y} (\hat X; \hat Y, \beta)} \\
&\geq 0, \label{eq:-btilted2}
\end{align}
where \eqref{eq:-btilted2} holds because $X^\star$ achieves the supremum in the right side of \eqref{eq:C(b)maxj}. Assume for the moment that $P_{\bar Y} \ll P_{Y^\star}$. Lemma \ref{lemma:verdupertube} implies that  
$
 D(\hat Y \| Y^\star) = \smallo{\alpha}
$. Thus, supposing that  
$
\E{\jmath_{X; Y^\star} (\bar X; \bar Y, \beta)} > \E{\jmath_{X; Y^\star} (X^\star; Y^\star, \beta)}
$
would lead to a contradiction, since then the left side of \eqref{eq:-btilted1} would be negative for a  sufficiently small $\alpha$. 

To complete the proof of \eqref{eq:C(b)maxjstar}, it remains to show $P_{Y^\star}$ dominates all $P_{\bar Y}$ such that $P_{\bar X} \to P_{Y|X} \to P_{\bar Y}$. By contradiction, assume that  $P_{\bar X}$ and $\mathcal F \subseteq \mathcal Y$  are such that $P_{\bar Y}(\mathcal F) > P_{Y^\star}(\mathcal F) = 0$, and define the mixture $P_{\hat X}$ as in \eqref{eq:PhatX}. Note that 
\begin{align}
 D(P_{Y|X} \| P_{\hat Y} | P_{\bar X}) &\geq D( \bar Y \| \hat Y)\\
 &\geq D(1 \{ \bar Y \in \mathcal F\} \| 1 \{ \hat Y \in \mathcal F\} )\\
 &\geq P_{\bar Y}(\mathcal F) \log \frac {P_{\bar Y}(\mathcal F)} {P_{\hat Y}(\mathcal F)} \\
 &= P_{\bar Y}(\mathcal F) \log \frac 1 {\alpha} \label{eq:Ystarmajor1a}.
\end{align}
Furthermore, we have
\begin{align}
&~  \E{\jmath_{X; \hat Y} (\hat X; \hat Y, \beta)}  - \E{\jmath_{X; Y^\star} (X^\star; Y^\star, \beta)} \notag
\\
  =&~ 
\alpha \E{\jmath_{X; \hat Y} (\bar X; \bar Y, \beta)} + (1 - \alpha)  \E{\jmath_{X; \hat Y} ( X^\star; Y^\star, \beta)} 
\notag\\&~
- \E{\jmath_{X; Y^\star} (X^\star; Y^\star, \beta)}  \\
\geq &~
\alpha \E{\jmath_{X; \hat Y} (\bar X; \bar Y, \beta)}  - \alpha \E{\jmath_{X; Y^\star} (X^\star; Y^\star, \beta)} \label{eq:Ystarmajor1}\\
\geq &~ \alpha \Big (  P_{\bar Y}(\mathcal F) \log \frac 1 {\alpha} - \lambda^\star \E{\mathsf b(\bar X)} + \lambda^\star \beta 
\notag\\ &~
-  \E{\jmath_{X; Y^\star} (X^\star; Y^\star, \beta)} \Big ) 
\label{eq:Ystarmajor2}\\
>&~ 0 \label{eq:Ystarmajor4},
\end{align}
where \eqref{eq:Ystarmajor1} is due to $D( Y^\star \| \hat Y) \geq 0$, \eqref{eq:Ystarmajor2} invokes \eqref{eq:Ystarmajor1a},
and \eqref{eq:Ystarmajor4} holds for sufficiently small $\alpha$, thereby contradicting \eqref{eq:C(b)maxj}. We conclude that indeed  $P_{\bar Y} \ll P_{Y^\star}$.

To show \eqref{eq:C(b)jstarconditional}, define the following function of a pair of probability distributions on $\mathcal X$:
\begin{align}
  F(P_{X}, P_{\bar X}) &=  \E{ \jmath_{ X; \bar Y}(X;  Y, \beta)} - D(X \| \bar X) \\
  &= \E{ \jmath_{ X;  Y}(X;  Y, \beta)} - D(X \| \bar X) + D(Y \| \bar Y) \\
 &\leq \E{ \jmath_{ X;  Y}(X;  Y, \beta)} \label{eq:-btilted3},
 \end{align}
where \eqref{eq:-btilted3} holds by the data processing inequality for relative entropy. Since equality in \eqref{eq:-btilted3} is achieved by $P_{X} = P_{\bar X}$, $\mathbb C(\beta)$ can be expressed as the double maximization
\begin{equation}
\mathbb C(\beta) = \max_{P_{\bar X}} \max_{P_X} F(P_X, P_{\bar X}) \label{eq:doublemax}.
\end{equation}

To solve the inner maximization in \eqref{eq:doublemax}, we invoke Lemma \ref{lemma:verdu} with
\begin{equation}
g(x) = \E{\jmath_{X; \bar Y}(x; Y, \beta) | X = x} 
\end{equation}
to conclude that
\begin{equation}
 \max_{P_X} F(P_X, P_{\bar X}) = \log \E{\exp\left(  \E{\jmath_{X; \bar Y}(\bar X; \bar Y, \beta) | \bar X} \right) },
\end{equation}
which in the special case $P_{\bar X} = P_{X^\star}$ yields, using representation \eqref{eq:doublemax}, 
\begin{align}
\mathbb C(\beta ) &\geq \log \E{\exp\left(  \E{\jmath_{X; Y^\star}(X^\star; Y, \beta) | X^\star} \right) }\\
&\geq \E{ \jmath_{ X;  Y^\star}( X^\star;  Y^\star, \beta) } \label{eq:-btiltedJensen}\\
&= \mathbb C(\beta) \label{eq:-btilted4}
\end{align}
where \eqref{eq:-btiltedJensen} applies Jensen's inequality to the strictly convex function $\exp(\cdot)$, and \eqref{eq:-btilted4} holds by the assumption. We conclude that, in fact, \eqref{eq:-btiltedJensen} holds with equality,  which implies that 
$\E{\jmath_{X; Y^\star}(X^\star; Y, \beta) | X^\star}$ is almost surely constant, thereby showing \eqref{eq:C(b)jstarconditional}.

\section{Proof of Corollary \ref{cor:btiltedvar}}
\label{appx:btiltedvar}

To show \eqref{eq:varinfodensity}, we invoke \eqref{eq:btilted} to write, for any $x \in \mathcal X$, 
\begin{align}
&~ 
 \Var{ \jmath_{ X;  Y^\star}( X;  Y, \beta) |  X = x } 
\notag\\
=&~   
  \Var{ \imath_{ X;  Y^\star}( X;  Y)  - \lambda^\star\left( \mathsf b( X) - \beta\right) |  X = x }\\
=&~
 \Var{ \imath_{ X;  Y^\star}( X;  Y) |  X = x }.
\end{align}

To show \eqref{eq:varbtilted}, we invoke \eqref{eq:C(b)jstarconditional} to write
\begin{align}
&~
\E{ \Var{ \jmath_{ X;  Y^\star}( X;  Y, \beta) |  X } }
\notag\\
=&~ 
\E{ \left(\jmath_{ X;  Y^\star}( X;  Y, \beta) \right)^2 } 
\notag\\ 
- 
&~
\E{\left(
\E{\jmath_{ X;  Y^\star}( X;  Y, \beta) |  X}\right)^2} \\
=&~ \E{ \left(\jmath_{ X;  Y^\star}( X;  Y, \beta) \right)^2 }  - \mathbb C^2(\beta)\\
=&~ \Var{\jmath_{ X;  Y^\star}( X;  Y, \beta)}. 
\end{align}

\section{Auxiliary result on the minimization of the cdf of a sum of independent random variables}
\label{appx:auxiliary}
Let $\mathcal D$ is a metric space with metric $d: \mathcal D^2 \mapsto \mathbb R^+$. Let $W_i(z), i = 1, \ldots, n$  be independent random variables parameterized by $z \in \mathcal D$. Denote
\begin{align}
D_n(z) &= \frac 1 n \sum_{i = 1}^n \E{ W_i(z)} \label{eq:Dnz}, \\
V_n(z) &= \frac 1 n \sum_{i = 1}^n \Var{W_i(z)} \label{eq:Vnz}, \\
T_n(z) &= \frac 1 n \sum_{i = 1}^n \E{ |W_i(z) - \E{W_i(z)} |^3  } \label{eq:Tnz}. 
\end{align}

Let $\ell_1$, $\ell_2$, $\ell_3$, $L_1$, $L_2$, $F_1$, $F_2$, $V_{\min}$ and $T_{\max}$ be positive constants. We assume that there exist $z^\star \in \mathcal D$ and sequences $D_n^\star$, $V_n^\star$ such that for all $z \in \mathcal D$,
\begin{align}
 D^\star_n - D_n(z) 
 &\geq 
 \ell_1 d^2\left(z, z^\star \right) - \frac{\ell_2}{\sqrt n} d\left(z, z^\star \right) - \frac{\ell_3}{n} \label{eq:ell_1},\\
 D^\star_n - D_n(z^\star) 
 &\leq 
  \frac{L_1}{n} \label{eq:L_1},   \\
 \left | V_n(z)  - V_n^\star \right | &\leq F_1 d \left( z, z^\star \right) + \frac{F_2}{\sqrt n} \label{eq:F_1},\\
V_{\min} &\leq V_n(z),  \label{eq:V_min}\\
 T_n(z) &\leq T_{\max} .\label{eq:T_max}
\end{align}

\begin{thm} 
In the setup described above, under assumptions \eqref{eq:ell_1}--\eqref{eq:T_max}, for any $A> 0$, there exists a $K \geq 0$ such that, for all $\left|\Delta\right| \leq \delta_n$ (where $\delta_n$ is specified below) and all  sufficiently large $n$: 
\begin{enumerate}[1.]

\item If $\delta_n = \frac A {\sqrt n}$,
\begin{align}
 \min_{z \in \mathcal D} \Prob{\sum_{i = 1}^n W_i(z) \leq  n\left( D^\star_n - \Delta\right)} 
 &\geq   
 Q\left( \Delta \sqrt{\frac{n} { V^\star_n} }\right)
 -  \frac {K} {\sqrt n} \label{eq:minprob:1}.
\end{align}
\label{thm:minprob:1}

\item For $\delta_n = A \sqrt{\frac{\log n}{n}}$, 
\begin{align}
 \min_{z \in \mathcal D} \Prob{\sum_{i = 1}^n W_i(z) \leq  n\left( D^\star_n - \Delta\right)} 
 \geq&~   
 Q\left( \Delta \sqrt{\frac{n} { V^\star_n} }\right)
 \notag\\
 &~
 - K \sqrt{ \frac{\log n}{n} }\label{eq:minprob:2}.
\end{align}
\label{thm:minprob:2}

\apxonly{
\item If the following tighter version of \eqref{eq:ell_1} and \eqref{eq:F_1} holds:
\begin{align}
 D^\star_n - D_n(z) 
 &\geq 
 \ell_1 d^2\left(z, z^\star \right) - \frac{\ell_3}{n} \label{eq:ell_1tighter}\\
 \left | V_n(z)  - V_n^\star \right | &\leq F_1 d \left( z, z^\star \right) + \frac{F_2}{n} \label{eq:F_1tighter}
\end{align}

and $\delta_n = 2 \ell_1 T_{\max}^{\frac 1 3} V_{\min}^{\frac 5 2}F_1^{-2}$, then
\begin{align}
 \min_{z \in \mathcal D} \Prob{\sum_{i = 1}^n W_i \leq  n\left(D^\star_n - \Delta\right)| Z = z} 
 &\geq   
 Q\left( \Delta \sqrt{\frac{n} { V^\star_n} }\right)
 -  \frac {K} {\sqrt n} \label{eq:minprob:3}
\end{align}
\label{thm:minprob:3}
}

\item Fix $0 \leq \beta \leq \frac 1 6$. If in \eqref{eq:F_1}, $V_n^\star = 0$ (which implies that $V_{\min} = 0$ in \eqref{eq:V_min}, i.e. we drop the requirement in Theorems \ref{thm:minprob}.\ref{thm:minprob:1} and  \ref{thm:minprob}.\ref{thm:minprob:2} that $V_{\min}$ be positive), then there exists $K \geq 0$ such that
for all $\Delta > \frac{A}{n^{\frac 1 2 + \beta}}$, where $A > 0$ is arbitrary
\begin{align}
 \min_{z \in \mathcal D} \Prob{\sum_{i = 1}^n W_i(z) \leq  n\left(D^\star_n + \Delta\right)} 
 &\geq   
1
 -  \frac{K}{A^{\frac 3 2}} \frac 1 {n^{\frac 1 4 - \frac 3 2 \beta }} \label{eq:minprob:4}.
\end{align}
\label{thm:minprob:4}

\end{enumerate}

 \label{thm:minprob}
 \end{thm}

Theorem \ref{thm:minprob} gives a general result on the minimization of a cdf of a sum of independent random variables parameterized by elements of a metric space:  it says that the minimum is approximately achieved by the sum with the largest mean, under regularity conditions. The metric nature of the parameter space is essential in making sure the means and the variances of $W_i(\cdot)$ behave like continuous functions: assumptions \eqref{eq:F_1} and \eqref{eq:L_1}  essentially ensure that functions $D_n(\cdot)$ and $D_n(z)$ are well-behaved in the neighborhood of the optimum, while assumption \eqref{eq:ell_1} guarantees that $D_n(\cdot)$ decays fast enough near its maximum.

Before we proceed to prove Theorem \ref{thm:minprob}, we recall the Berry-Esseen refinement of the central limit theorem. 

\begin{thm}[{Berry-Esseen CLT, e.g. \cite[Ch. XVI.5 Theorem 2]{feller1971introduction}}]
\label{thm:Berry-Esseen}
Fix a positive integer $n$. Let $W_i$, $i = 1, \ldots, n$ be independent. Then, for any real $t$
\begin{equation}
\left| \mathbb P \left[ \sum_{i = 1}^n W_i > n \left( D_n + t \sqrt {\frac{V_n}{ n}}\right) \right]  - Q(t) \right| \leq \frac {B_n}{\sqrt n},
\label{eq:BerryEsseen}
\end{equation}
where
\begin{align}
D_n &= \frac 1 n \sum_{i = 1}^n \E{ W_i} \label{eq:BerryEsseenDn}, \\
V_n &= \frac 1 n \sum_{i = 1}^n \Var{W_i} \label{eq:BerryEsseenVn}, \\
T_n &= \frac 1 n \sum_{i = 1}^n \E{ |W_i - \E{W_i} |^3 } \label{eq:BerryEsseenTn}, \\
B_n &=  \frac{c_0 T_n}{V_n^{3/2}} \label{eq:BerryEsseenBn},
\end{align}
and $0.4097 \leq c_0 \leq 0.5600$ ($c_0 \leq 0.4784$ for identically distributed $W_i$). 
\end{thm}

We also make note of the following lemma, which deals with the behavior of the $Q$-function. 

\begin{lemma}[{\cite[Lemma 4]{kostina2012jscc}}]
Fix $b \geq 0$. Then, there exists $q \geq 0$ such that for all $z \geq -\frac{1}{2 b}$ and all $n \geq 1$, 
\begin{equation}
  Q \left(\sqrt n z \right)
 -
 Q \left( \sqrt n z \left( 1 + b z \right) \right)
 \leq
  \frac{q}{\sqrt n} \label{eq:Qbound}.
\end{equation}
\label{lemma:Qbound}
\end{lemma}

We are now equipped to prove Theorem \ref{thm:minprob}.

\begin{proof}[Proof of Theorem \ref{thm:minprob}]
To show \eqref{eq:minprob:4}, denote for brevity $\zeta =d\left(z, z^\star \right)$ and write
\begin{align}
&~ 
\Prob{ \sum_{i = 1}^n W_i(z) > n\left( D_n^\star + \Delta \right) } \notag\\
\leq &~ \Prob{ \sum_{i = 1}^n W_i(z) > n\left( D_n(z) + \ell_1 \zeta^2 - \frac{\ell_2}{\sqrt n} \zeta - \frac{\ell_3}{n}  + \frac{A}{n^{\frac 1 2 + \beta}} \right)} \label{eq:-minprobCheb1}\\
\leq &~ \frac{1}{n} \frac{F_1 \zeta + \frac {F_2}{\sqrt n}}{\left( \ell_1 \zeta^2 - \frac{\ell_2}{\sqrt n} \zeta - \frac{\ell_3}{n}  + \frac{A}{n^{\frac 1 2 + \beta} }\right)^2} \label{eq:-minprobCheb2} \\
\leq &~ \frac{K} {A^{\frac 3 2}} \frac 1 {n^{\frac 1 4 - \frac 3 2 \beta }} \label{eq:-minprobCheb3},
\end{align}
 where
\begin{itemize}
 \item \eqref{eq:-minprobCheb1} uses \eqref{eq:ell_1} and the assumption on the range of $\Delta$;
 \item \eqref{eq:-minprobCheb2} is due to Chebyshev's inequality and $V_n^\star = 0$;
 \item  \eqref{eq:-minprobCheb3} is by a straightforward algebraic exercise revealing that $\zeta$ that maximizes the left side of \eqref{eq:-minprobCheb3} is proportional to $\frac{A^{\frac 1 2}}{n^{\frac 1 4 + \frac 1 2 \beta}}$.
\end{itemize}

We proceed to show \eqref{eq:minprob:1} and \eqref{eq:minprob:2}. 

 Denote
\begin{align}
 g_n(z) &= \Prob{\sum_{i = 1}^n W_i(z) \leq n( D^\star_n - \Delta)}.
 \end{align}
 
Using \eqref{eq:V_min} and \eqref{eq:T_max}, observe
\begin{align}
 \frac {c_0T_n(z)}{V_n^\frac 3 2(z)} \leq B  
 = \frac{ c_0 T_{\max} }{V_{\min}^{\frac 3 2}} < \infty.
 \end{align}
 Therefore the Berry-Esseen bound yields:
\begin{equation}
 \left| g_n(z) - Q\left( \sqrt{n} \nu_n(z)\right)\right| \leq \frac {B}{\sqrt n} \label{eq:-be},
\end{equation}
where
\begin{equation}
 \nu_n(z) \triangleq \frac{D_n(z) -  D^\star_n + \Delta }{\sqrt{V_n(z)}}.
 \end{equation}
 
 Denote 
\begin{equation}
\nu_n^\star \triangleq  \frac{\Delta}{\sqrt { V^\star_n}} 
 \end{equation}

Since 
\begin{align}
g_n(z)  
&=
  Q( \sqrt{n} \nu_n^\star)
  + \left[ g_n(z) - Q\left( \sqrt{n} \nu_n(z)\right) \right]
 \notag\\
  &
  + \left[ Q\left( \sqrt n \nu_n(z)\right) - Q( \sqrt n \nu_n^\star ) \right]\\
&\geq Q(\sqrt n \nu_n^\star) - \frac{B}{\sqrt n} + \left[ Q\left( \sqrt n \nu_n(z)\right) - Q(\sqrt n \nu_n^\star) \right],
\end{align}
 to show  \eqref{eq:minprob:1}, 
 it suffices to show that 
\begin{equation}
 Q(\sqrt n \nu_n^\star) - \min_{z \in \mathcal D} Q\left( \sqrt n \nu_n(z)\right) \leq \frac {q}{\sqrt n} \label{eq:lemmafg2a}
\end{equation}
for some $q  \geq 0$, and to show \eqref{eq:minprob:2},  replacing $q$ with $q \sqrt{\log n}$ in the right side of \eqref{eq:lemmafg2a} would suffice.

Since $Q$ is monotonically decreasing, to achieve the minimum in  \eqref{eq:lemmafg2a} we need to maximize $\sqrt n \nu_n(z)$.  
As will be proven shortly, for appropriately chosen $a,b, c > 0$ we can write
\begin{align}
\max_{z \in \mathcal D}\nu_n(z) 
&\leq
\nu_n^\star + b \nu_n^{\star 2} + \frac{c \delta_n}{\sqrt n}  \label{eq:numaxupper}
\end{align}
for $n$ large enough. 

If
\begin{equation} 
  \Delta \geq - \frac{\sqrt {V_{\min}}}{2b} = - A,
\end{equation}
then $\nu_n^\star \geq - \frac 1 {2b}$,  and Lemma \ref{lemma:Qbound} applies to $\nu_n^\star$. So, using   \eqref{eq:numaxupper}, the fact that $Q(
\cdot)$ is monotonically decreasing and  Lemma \ref{lemma:Qbound}, we conclude that there exists $q > 0$ such that 
\begin{align}
 &~ 
 Q \left(\sqrt n \nu_n^\star \right) - \min_{z \in \mathcal D} Q\left( \sqrt n \nu_n(z) \right) 
 \notag\\
 \leq &~
 Q \left(\sqrt n \nu_n^\star \right)
 -
 Q \left(  \sqrt n \nu_n^\star + \sqrt n b \nu_n^{\star 2} + c \delta_n \right)\\
 \leq &~ 
  Q \left(\sqrt n \nu_n^\star \right)
 -
 Q \left(  \sqrt n \nu_n^\star + \sqrt n b \nu_n^{\star 2} \right)
 +
 \frac{c}{\sqrt{2 \pi}} \delta_n   \label{eq:-Qbounda} \\
 \leq &~ \frac{q}{\sqrt n} +  \frac{c}{\sqrt{2 \pi}} \delta_n  \label{eq:-Qbound},
 \end{align}
 where
\begin{itemize}
 \item \eqref{eq:-Qbounda} is due to 
\begin{equation}
Q(z + \xi)  \geq Q(z) - \frac{\xi}{\sqrt{2 \pi}},
\end{equation}
 which holds for arbitrary $z$ and $\xi \geq 0$, 
 \item \eqref{eq:-Qbound} holds by Lemma \ref{lemma:Qbound} as long as $\nu_n^\star \geq - \frac{1}{2b}$. 
\end{itemize}
Thus, \eqref{eq:-Qbound} establishes  \eqref{eq:minprob:1} and \eqref{eq:minprob:2}. 
It remains to prove \eqref{eq:numaxupper}. 
To upper-bound $\max_{z \in \mathcal D} \nu_n(z)$, denote for convenience
\begin{align}
 f_n(z) &=  \frac{D_n(z) -  D^\star_n}{ \sqrt{V_n(z)} },\\
 g_n(z) &= \frac{1}{ \sqrt{V_n(z)} },
\end{align}
and note, using \eqref{eq:ell_1}, \eqref{eq:L_1}, \eqref{eq:V_min}, \eqref{eq:T_max} and (by H\"older's inequality)
\begin{equation}
 V_n(z) \leq T_{\max}^{\frac 2 3} \label{eq:tmaxholder},
\end{equation}
that
\begin{align}
 f_n(z^\star)  - f_n(z) &= \frac{D_n(z^\star) - D_n^\star}{\sqrt{V_n(z^\star)}}
 -
 \frac{D_n(z) - D_n^\star}{\sqrt{V_n(z)}}
 \\
&\geq  \ell_1^\prime d^2(z, z^\star) -  \frac{\ell_2^\prime}{\sqrt n} d(z, z^\star) -  \frac {\ell_3^\prime} n  \label{eq:fn},
\end{align}
where
\begin{align}
\ell_1^\prime &=  T_{\max}^{-\frac 1 3} \ell_1,\\
\ell_2^\prime &= V_{\min}^{- \frac 1 2} \ell_2, \\
\ell_3^\prime &= V_{\min}^{- \frac 1 2}(L_1 + \ell_3).
\end{align}

Observe that for $a, b > 0$
\begin{equation}
\left| \frac 1 {\sqrt a} - \frac 1 {\sqrt b}\right| \leq \frac{\left| a - b\right| }{2\min \left\{ a, b\right\}^{\frac 3 2 }} \label{eq:fact1}, 
\end{equation}
so, using \eqref{eq:F_1} and \eqref{eq:V_min}, we conclude
\begin{equation}
\left| \frac 1 {\sqrt {V_n(z)}} - \frac 1 {\sqrt{V_n^\star}}\right| \leq 
 F_1^\prime d(z, z^\star) + \frac {F_2^\prime}{\sqrt n} \label{eq:gn}, 
\end{equation}
where
\begin{align}
F_1^\prime &=  \frac 1 2V_{\min}^{- \frac 3 2}  F_1,\\
F_2^\prime &= \frac 1 2 V_{\min}^{- \frac 3 2} F_2.
\end{align}
 
 Let $z_0$ achieve the maximum $\max_{z \in \mathcal D} \nu_n(z)$, i.e.
\begin{equation}
 \max_{z \in \mathcal D} \nu_n(z) = f_n(z_0) + \Delta g_n(z_0).
\end{equation}

Using \eqref{eq:gn} and \eqref{eq:fn}, we have, 
\begin{align}
&~
 \nu_n(z_0) - \nu_n(z^\star) 
 \notag\\
 =  
 &~
 \left( f_n(z_0) - f_n(z^\star)\right) + \Delta \left( g_n(z_0) - g_n(z^\star)\right) \\
 \leq
 &~
  - \ell_1^\prime d^2(z_0, z^\star) 
 + \left( \frac{\ell_2^\prime}{\sqrt n} +  |\Delta| F_1^\prime \right) d(z_0, z^\star)
 +   \frac{2 F_2^\prime |\Delta|}{\sqrt n} 
  \notag\\
  + 
  &~
   \frac {\ell_3^\prime} n 
 \\
 \leq
 &~
  \frac 1 {4 \ell_1^\prime} \left(  \frac{\ell_2^\prime}{\sqrt n} +  |\Delta| F_1^\prime \right)^2
 +   \frac{2 F_2^\prime |\Delta|}{\sqrt n} 
  +  \frac {\ell_3^\prime} n,
 \label{eq:numax}
\end{align}
where \eqref{eq:numax} follows because the maximum of its left side is achieved at 
$d(z_0, z^\star) = \frac{1}{2 \ell_1^\prime}\left( \frac{\ell_2^\prime}{\sqrt n} +  |\Delta| F_1^\prime \right)$. 
Using \eqref{eq:ell_1}, \eqref{eq:V_min}, \eqref{eq:gn}, we upper-bound
\begin{equation}
\nu_n(z^\star) \leq \nu_n^\star  + \frac{F_2^\prime|\Delta|}{\sqrt n} + \frac{\ell_3}{nV_{\min}} + \frac{\ell_3 F_2^\prime}{n^{\frac 3 2 }} \label{eq:nustar_ub}.
\end{equation}
Applying \eqref{eq:numax} and \eqref{eq:nustar_ub} to upper-bound 
$\max_{z \in \mathcal D} \nu_n(z)$, we have established \eqref{eq:numaxupper} in which
\begin{equation}
b = \frac{F_1^{\prime 2} T_{\max}^{\frac 2 3}}{4 \ell_1^\prime}, 
\end{equation}
where we used \eqref{eq:F_1} and \eqref{eq:tmaxholder} to upper-bound $\Delta^2 = \nu_n^{\star 2} V_n^\star$,  thereby completing the proof.

\end{proof}

\section{Proof of the converse part of Theorem \ref{thm:2orderCost}}
\label{appx:2orderCost}


Given a finite set $\mathcal{A}$, let $\mathcal P$ be the set of all distributions on $\mathcal A$ that satisfy the cost constraint, 
\begin{equation}
 \E{\mathsf b(\mathsf X)} \leq \beta \label{eq:Eb(X)<=b},
\end{equation}
which is a convex set in $\mathbb R^{|\mathcal A|}$. 

Leveraging an idea of Tomamichel and Tan \cite{tomamichel2013thirdorder}, we will weaken \eqref{eq:Ccosta}  by choosing $P_{\bar Y^n}$ to be a  convex combination of non-product distributions with weights chosen to favor those distributions that are close to $P_{Y^{\star n}}$. Specifically (cf. \cite{tomamichel2013thirdorder}), 
\begin{equation}
P_{\bar Y^n}(y^n) = \frac 1 {A} \sum_{\mathbf k \in \mathcal K} \exp \left( - |\mathbf k|^2\right) \prod_{i = 1}^n P_{ \mathsf Y | \mathbf K = \mathbf k }(y_i)  \label{eq:PYbar},
\end{equation}
where 
$\{P_{ \mathsf Y | \mathbf K = \mathbf k }, ~ \mathbf k \in \mathcal K\}$
 are defined as follows, for some $c > 0$, 
\begin{align}
 &~P_{\mathsf Y | \mathbf K = \mathbf k }(\mathsf y) = P_{\mathsf Y^\star}( \mathsf y) + \frac{k_{\mathsf y}} {\sqrt{n c}} \label{eq:PYk}, \\
 &\mathcal K = \Bigg\{ \mathbf k \in \mathbb Z^{|\mathcal B|} \colon \sum_{\mathsf y \in \mathcal B} k_{\mathsf y} = 0, 
 \notag\\
&~
- P_{\mathsf Y^\star}( \mathsf y) + \frac{1}{\sqrt {n c}} \leq \frac{k_{\mathsf y}}{\sqrt{n c}} 
 \leq  1- P_{\mathsf Y^\star}( \mathsf y)
 \Bigg\}, \label{eq:kset}\\
 &~A = \sum_{\mathbf k \in \mathcal K}  \exp \left( - |\mathbf k|^2\right) < \infty.
\end{align}
 
Denote by $P_{ \Pi(\mathsf Y) }$ the minimum Euclidean distance approximation of an arbitrary $P_{\mathsf Y} \in \mathcal Q$, where $\mathcal Q$ is the set of distributions on the channel output alphabet $\mathcal B$, in the set 
$\left\{P_{\mathsf Y | \mathbf K = \mathbf k}\colon ~ \mathbf k \in \mathcal K \right\}$:
\begin{equation}
 P_{\Pi(\mathsf Y)} = P_{\mathsf Y | \mathbf K = \mathbf k^\star}  \text{ where } 
 \mathbf k^\star = \arg\min_{\mathbf k \in \mathcal K}  \left|P_{\mathsf Y} - P_{\mathsf Y | \mathbf K = \mathbf k} \right| \label{eq:PYhatdef}.
\end{equation}
The quality of approximation \eqref{eq:PYhatdef} is governed by \cite{tomamichel2013thirdorder}
\begin{equation}
\left|  P_{ \Pi(\mathsf Y) } -  P_{\mathsf Y} \right| \leq \sqrt{\frac{|\mathcal B| (|\mathcal B| - 1 )}{n c}} \label{eq:PYhat}.
\end{equation}
We say that $x^n \in \mathcal A^n$ has type $P_{\hat {\mathsf X}}$ if the number of times each letter $a \in \mathcal A$ is encountered in $x^n$ is $n P_{\mathsf X}(a)$. An $n$-type is a distribution whose masses are multiples of $\frac{1}{n}$. Denote by $P_{\hat{\mathsf X}}$ the minimum Euclidean distance approximation of $P_{\mathsf X}$ in the set of $n$-types, that is,  
\begin{equation}
P_{\hat{\mathsf X}} = \arg\min_{ \substack{ P \in \mathcal P \colon \\
P \text{ is an $n$-type }}
 } \left|P_{\mathsf X} - P\right| \label{eq:-Pi}.
\end{equation}
The accuracy of approximation in \eqref{eq:-Pi} is controlled by the following inequality: 
\begin{equation}
  \left| P_{\mathsf X} - P_{\hat {\mathsf X}} \right| \leq \frac {\sqrt{|\mathcal A| \left( |\mathcal A| - 1\right)}}{n} \label{eq:-minEuclid}.
\end{equation}

For each $P_{\mathsf X} \in \mathcal P$, let $x^n \in \mathcal A^n$ be an arbitrary sequence of type $P_{\hat {\mathsf X}}$, and lower-bound the sum in \eqref{eq:PYbar} by the term containing $P_{\Pi({\mathsf Y})}$ to obtain: 
\begin{align}
 \jmath_{X^n; \bar Y^n}(x^n; y^n, \beta) &\leq 
 \sum_{i = 1}^n \jmath_{\mathsf X; \Pi(\mathsf Y) } \left( x_i, y_i, \beta \right) 
 \notag\\
 &
 + nc \left|  P_{ \Pi(\mathsf Y) } - P_{\mathsf Y^\star} \right|^2 + A \label{eq:-Cj1} .
\end{align}

Applying \eqref{eq:PYbar} and \eqref{eq:-Cj1} to loosen \eqref{eq:Ccosta}, we conclude by Theorem \ref{thm:Ccost} that, as long as an $(n, M, \epsilon^\prime)$ code  exists, for an arbitrary $\gamma > 0$, 
\begin{equation}
\epsilon^\prime \geq \min_{P_{\mathsf X} \in \mathcal P} 
\Prob{\sum_{i = 1}^n W_i(P_{\mathsf X})
\leq \log M - \gamma - A } - \exp\left( -\gamma\right)  \label{eq:-Ca},
\end{equation}
where
\begin{align}
W_i(P_{\mathsf X})  &= \jmath_{\mathsf X; \Pi(\mathsf Y) } \left( x_i, Y_i, \beta \right) + c\left|  P_{ \Pi(\mathsf Y) } - P_{\mathsf Y^\star} \right|^2  \label{eq:Wi1},
\end{align}
and $Y_i$ is distributed according to $P_{\mathsf Y | \mathsf X = x_i}$.\footnote{Strictly speaking, the order of $W_i(P_{\mathsf X})$, $i = 1, \ldots, n$ depends on the particular choice of sequence $x^n$ of type $P_{\hat {\mathsf X}}$. However, since the distribution of the sum $\sum_{i = 1}^n W_i(P_{\mathsf X})$ does not depend on their relative order, we may choose this sequence arbitrarily.}
To evaluate the minimization on the right side of \eqref{eq:-Ca}, we will apply Theorem \ref{thm:minprob} with $\mathcal D = \mathcal P$, $z = P_{\mathsf X}$, $z^\star = P_{\mathsf X^\star}$, $W_i(\cdot)$ in \eqref{eq:Wi1}, and the metric being the usual Euclidean distance in $\mathbb R^n$. 

Define the following functions $\mathcal P \times \mathcal Q \mapsto \mathbb R_+$:
 \begin{align}
  D(P_{\mathsf X}, P_{\bar{\mathsf Y}}) &=  \E{\jmath_{\mathsf X; {\bar{\mathsf Y}}}(\mathsf X; \mathsf Y, \beta)}  \label{eq:Dhat} + c\left|  P_{ \bar{\mathsf Y} } - P_{\mathsf Y^\star} \right|^2,
\\
  V (P_{\mathsf X}, P_{\bar{\mathsf Y}}) &= \E{\Var{\jmath_{\mathsf X; {\bar{\mathsf Y}}}(\mathsf X; \mathsf Y, \beta) \mid \mathsf X} } \label{eq:Vhat}, \\
 T (P_{\mathsf X}, P_{\bar{\mathsf Y}}) &= \E{ \left| \jmath_{\mathsf X; {\bar{\mathsf Y}}}(\mathsf X; \mathsf Y, \beta) - \E{\jmath_{\mathsf X; \bar{\mathsf Y}}(\mathsf X; \mathsf Y, \beta) | \mathsf X}\right|^3 } \label{eq:That},
 \end{align}
 where the expectations are with respect to $P_{\mathsf Y | \mathsf X} P_{\mathsf X}$. 
 
With the choice in \eqref{eq:Wi1} the functions \eqref{eq:Dnz}--\eqref{eq:Tnz} are particularized to the following mappings $\mathcal P\mapsto \mathbb R_+$: 
  \begin{align}
D_n(P_{\mathsf X}) &=   D \left(  {P_{\hat{\mathsf X}}},  P_{\Pi( {\mathsf Y}) } \right) \label{eq:Dnz1},\\
V_n (P_{\mathsf X}) &=   V \left( P_{\hat{\mathsf X}}, P_{\Pi( {\mathsf Y}) } \right)\label{eq:Vnz1}, \\
T_n (P_{\mathsf X}) &=  T \left( P_{\hat{\mathsf X}}, P_{\Pi( {\mathsf Y}) }  \right).
 \label{eq:Tnz1}
 \end{align}
and $D_n^\star$, $V_n^\star$ are
\begin{align}
D_n^\star &= C(\beta),\\
V_n^\star &= V(\beta).
\end{align}
 
We perform the minimization on the right side of \eqref{eq:-Ca} separately for 
$P_{\mathsf X} \in \mathcal P_{\delta}^\star$ 
and 
$P_{\mathsf X} \in \mathcal P \backslash \mathcal P_{\delta}^\star$,
where
\begin{align}
\mathcal P_{\delta}^\star &= \left\{ P_{\mathsf X} \in \mathcal P \colon 
\left| P_{\mathsf X} - P_{\mathsf X^\star} \right| \leq \delta \right\}.  
\end{align}
Assuming without loss of generality that all outputs in $\mathcal B$ are accessible (meaning that for each $\mathsf y \in\mathcal B$, there exists $\mathsf x \in \mathcal A$ with $P_{\mathsf Y | \mathsf X} (\mathsf y | \mathsf x) > 0$; this implies in particular that $P_{\mathsf Y^\star}(\mathsf y) > 0$ for all $\mathsf y \in \mathcal B$), we choose $\delta > 0$ so that 
\begin{align}
\min_{P_{\mathsf X} \in \mathcal P_{\delta}^\star}\min_{\mathsf y \in \mathcal B} P_{\mathsf Y}(\mathsf y)   &= p_{\min} > 0 \label{P_Ymin}, \\
 2 \min_{P_{\mathsf X} \in \mathcal P_\delta^\star} V\left( P_{\mathsf X} \right)  &\geq V(\beta). 
 \end{align}
 
   To perform the minimization on the right side of \eqref{eq:-Ca} over $\mathcal P_{\delta}^\star$, we will invoke Theorem \ref{thm:minprob} with $\mathcal D = \mathcal P_{\delta}^\star$, the metric being the usual Euclidean distance between $|\mathcal A|$-vectors. Let us check that the assumptions of Theorem \ref{thm:minprob} are satisfied. It is easy to verify directly that the functions 
 $P_{\mathsf X} \mapsto D(P_{\mathsf X}, P_{\mathsf Y})$, 
 $P_{\mathsf X} \mapsto  V (P_{\mathsf X}, P_{\mathsf Y}) $,
  $P_{\mathsf X} \mapsto T (P_{\mathsf X}, P_{\mathsf Y}) $
  are continuous (and therefore bounded) on $\mathcal P$ and infinitely differentiable on $\mathcal P_{\delta}^\star$. Therefore, assumptions \eqref{eq:V_min} and \eqref{eq:T_max} of Theorem \ref{thm:minprob} are met. To verify that \eqref{eq:ell_1} holds, write, for $\zeta = |P_{\mathsf X} - P_{\mathsf X^\star}|$,

\begin{align}
C(\beta) - D \left(  {P_{\hat{\mathsf X}}},  P_{\Pi( {\mathsf Y}) } \right)
 =&~ 
 C(\beta) - D \left(  P_{\mathsf X}, P_{\mathsf Y} \right) -   \frac{\ell_2}{\sqrt n} \zeta - \frac{\ell_3}{n}   \label{eq:ass1a} \\
 \geq&~
 \ell_1 \zeta^2 -
\frac{\ell_2}{\sqrt n} \zeta - \frac{\ell_3}{n} , 
 \label{eq:ass1f}
\end{align}
where all constants $\ell_1$, $\ell_2$, $\ell_3$ are positive, and:
\begin{itemize}
\item  to show \eqref{eq:ass1a}, observe that for a fixed $P_{\bar {\mathsf Y}}$, $D \left(  \cdot,  P_{\bar {\mathsf Y} } \right)$ is a linear function of $P_{\mathsf X}$, so in view of \eqref{eq:-minEuclid}
\begin{align}
 \left| D \left(  {P_{\hat {\mathsf X}}},  P_{\Pi( {\mathsf Y}) } \right) - D \left(  {P_{{\mathsf X}}},  P_{\Pi( {\mathsf Y}) } \right) \right| \leq  \frac{L_1}{n} \label{eq:Dlinear}.
\end{align}
Furthermore, 
\begin{align}
 &~
 D \left(  {P_{{\mathsf X}}},  P_{\Pi( {\mathsf Y}) } \right) 
 \notag\\
 =&~ D \left(  P_{\mathsf X}, P_{\mathsf Y} \right) +  c  |P_{\Pi({\mathsf Y})} - P_{\mathsf Y^\star}|^2 - c |P_{{\mathsf Y}} - P_{\mathsf Y^\star}|^2 
 \notag\\
  + 
&~  
  D(P_{\mathsf Y} \| P_{ \Pi(\mathsf Y)})\\
 \leq
 &~
   D \left(  P_{\mathsf X}, P_{\mathsf Y} \right) +  c  |P_{\Pi({\mathsf Y})} - P_{\mathsf Y}|^2
 \notag\\
   +
   &~
   2 c |P_{{\mathsf Y}} - P_{\mathsf Y^\star}| |P_{\Pi({\mathsf Y})} - P_{\mathsf Y}| + D(P_{\mathsf Y} \| P_{ \Pi(\mathsf Y)})\\
 \leq&~ D \left(  P_{\mathsf X}, P_{\mathsf Y} \right) + \frac{\ell_2}{\sqrt n} \zeta + \frac{\ell_{3}^{ \prime}}{n}   \label{eq:ass1a-},
\end{align}
where we used the triangle inequality,  \eqref{eq:PYhat},  a ``reverse Pinsker inequality'' \cite[Lemma~6.3]{csiszar2006context}:
\begin{align}
 D(\mathsf Y \| \bar {\mathsf Y}) 
 &\leq \frac{\log e}{\min_{b \in \mathcal B} P_{\bar{\mathsf Y}}(b)} \left|  P_{ \mathsf Y} - P_{\bar{\mathsf Y}} \right|^2 \label{eq:divergence_ub}
\end{align}
and
\begin{equation}
 |P_{\mathsf Y} - P_{\bar{\mathsf Y}}| \leq |P_{\mathsf Y | \mathsf X}| |P_{\mathsf X} - P_{\bar{\mathsf X}}|, \label{eq:spectralnorm}
\end{equation}
  where $P_{\bar {\mathsf X}} \to P_{\mathsf Y | \mathsf X} \to P_{\bar{\mathsf Y}}$, and the spectral norm of $P_{\mathsf Y | \mathsf X}$ satisfies $ | P_{\mathsf Y | \mathsf X} | \leq \sqrt{|\mathcal A|}$.

\item \eqref{eq:ass1f} uses
\begin{equation}
 \E{\jmath_{\mathsf X; {\mathsf Y}}(\mathsf X; \mathsf Y, \beta)} \leq C(\beta) - \ell_1^\prime \zeta^2 \label{eq:quadraticdecay},
\end{equation}
where $\ell_1^\prime >0$, and
\begin{equation}
 \ell_1 = \ell_1^\prime - c |\mathcal A|
\end{equation}
can be made positive for a small enough $c$. Inequality \eqref{eq:quadraticdecay} can be shown following the reasoning in \cite[(497)--(505)]{polyanskiy2010channel} invoking \eqref{eq:C(b)jstarconditional} in lieu of the corresponding property for the conventional information density. Here we provide a simpler proof using Pinsker's inequality.  
Viewing $P_{\mathsf X}$ as a vector and $P_{\mathsf Y | \mathsf X}$ as a matrix, write
\begin{equation}
 P_{\mathsf X} = P_{\mathsf X^\star} + v_0 + v_{\perp},
\end{equation}
where $v_0$ and $v_{\perp}$ are projections of $P_{\mathsf X} - P_{\mathsf X^\star}$ onto $\mathrm{Ker} P_{\mathsf Y | \mathsf X}$ and $(\mathrm{Ker} P_{\mathsf Y | \mathsf X})^\perp$ respectively, where 
\begin{equation}
 \mathrm{Ker} P_{\mathsf Y | \mathsf X} = \left\{ v \in \mathbb R^{|\mathcal A|} \colon v^T P_{\mathsf Y | \mathsf X} = 0 \right\} .
\end{equation}
We consider two cases $v_{\perp} = 0$ and $v_{\perp} \neq 0$ separately. Condition $v_{\perp} = 0$ implies $P_{\mathsf X} \to P_{\mathsf Y | \mathsf X} \to P_{\mathsf Y^\star}$, which combined with $P_{\mathsf X} \neq P_{\mathsf X^\star}$ and \eqref{eq:C(b)jstarconditional} means that the complement of $F =\mathrm{supp} (P_{\mathsf X^\star})$ is nonempty and 
\begin{equation}
 a \triangleq C(\beta) -  \max_{x \notin F} \E{\jmath_{\mathsf X; \mathsf Y^\star}(\mathsf x; \mathsf Y, \beta) | \mathsf X = \mathsf x }
 \end{equation} 
is positive. 
Therefore
\begin{align}
&~
 \E{\jmath_{\mathsf X; \mathsf Y}(\mathsf X; \mathsf Y, \beta)} 
\notag\\
=&~ \E{\jmath_{\mathsf X; \mathsf Y^\star}(\mathsf X; \mathsf Y, \beta)} 
\\
 =&~ \E{\jmath_{\mathsf X; \mathsf Y^\star}(\mathsf X; \mathsf Y, \beta), \mathsf X \in F } 
 + \E{\jmath_{\mathsf X; \mathsf Y^\star}(\mathsf X; \mathsf Y, \beta), \mathsf X \notin F } \\
 &\leq C(\beta) P_{\mathsf X} \left(F \right) + P_{\mathsf X} \left(F^c \right) (C(\beta) - a) \label{eq:ass12a}\\
 &\leq C(\beta)  - (\lambda_{\min}^+(P_F^2))^{1/2} a |v|\\
 &\leq C(\beta)  - \frac{1}{4} (\lambda_{\min}^+(P_F^2))^{1/2} a |v|^2,
\end{align}
where \eqref{eq:ass12a} uses \eqref{eq:C(b)jstarconditional},  $P_F$ is the orthogonal projection matrix onto $F^c$ and $\lambda_{\min}^+ (\cdot)$ is the minimum nonzero eigenvalue of the indicated positive semidefinite matrix. 

If $v_{\perp} \neq 0$, write
\begin{align}
&~
\E{\jmath_{\mathsf X; \mathsf Y}(\mathsf X; \mathsf Y, \beta)} 
\notag\\
=&~ \E{\jmath_{\mathsf X; \mathsf Y^\star}(\mathsf X; \mathsf Y, \beta)} - D(P_{\mathsf Y} \| P_{\mathsf Y^\star})\\
 &\leq \E{\jmath_{\mathsf X; \mathsf Y^\star}(\mathsf X; \mathsf Y, \beta)} - \frac 1 2 \left| P_{\mathsf Y} - P_{{\mathsf Y}^\star} \right|^2 \log e \label{eq:pinsker}\\
 &\leq C(\beta) - \frac 1 2 \left| P_{\mathsf Y} - P_{{\mathsf Y}^\star} \right|^2  \log e \label{eq:ass11a},
\end{align}
where \eqref{eq:pinsker} is by Pinsker's inequality, and \eqref{eq:ass11a} is by \eqref{eq:C(b)maxjstar}. To conclude the proof of \eqref{eq:quadraticdecay}, we lower bound the second term in \eqref{eq:ass11a} as follows. 
\begin{align}
\left| P_{\mathsf Y} - P_{{\mathsf Y}^\star} \right|^2 &= 
\left| \left( P_{\mathsf X} - P_{\mathsf X^\star}\right)^T P_{\mathsf Y | \mathsf X} \right|^2  \\
&= \left| v^T_\perp P_{\mathsf Y | \mathsf X} \right|^2\\
&\geq \lambda_{\min}(P_{\mathsf Y | \mathsf X}) |v_\perp|^2\\
&\geq \lambda_{\min}^+ (P_{\mathsf Y | \mathsf X} P_{\mathsf Y | \mathsf X}^T ) \lambda_{\min}^+ (P^2_\perp)  |v|^2,
\end{align}
where $P_{\perp}$ is the orthogonal projection matrix onto $(\mathrm{Ker} P_{\mathsf Y | \mathsf X})^\perp$.  

\end{itemize}

To establish \eqref{eq:L_1}, write
\begin{align}
C(\beta) - D(P_{\hat{\mathsf X}}, P_{\Pi( {\mathsf Y}) }) 
\leq&~
C(\beta) - D(P_{{\mathsf X}}, P_{\Pi( {\mathsf Y}) }) + \frac{L_1}{n} \label{eq:ass2}\\
\leq&~ C(\beta) -  \E{\jmath_{\mathsf X; {\mathsf Y}}(\mathsf X; \mathsf Y, \beta)} + \frac{L_1}{n} \label{eq:ass2a},
\end{align}
where \eqref{eq:ass2} is due to \eqref{eq:Dlinear}. Substituting $\mathsf X = \mathsf X^\star$ into \eqref{eq:ass2a}, we obtain \eqref{eq:L_1}. 

Finally, to verify \eqref{eq:F_1}, write
\begin{align}
&~
\left| V \left( P_{\hat{\mathsf X}}, P_{\Pi( {\mathsf Y}) } \right) - V(\beta) \right| 
\notag\\
\leq&~
\left| V(P_{\mathsf X}, P_{\mathsf Y}) - V(\beta) \right| 
+
\left| V(P_{\mathsf X}, P_{\mathsf Y}) - V \left( P_{\hat{\mathsf X}}, P_{ \mathsf Y }\right)\right| 
\notag\\
+
&~
\left| V \left( P_{\hat{\mathsf X}}, P_{\Pi( {\mathsf Y}) } \right) -  
V \left( P_{\hat{\mathsf X}}, P_{ {\mathsf Y} } \right)\right| \\
\leq&~
 F_1 |P_{\mathsf X} - P_{\mathsf X^\star}|
 +
 F_2^\prime |P_{\mathsf X} - P_{\hat{\mathsf X}}| + F_2^{\prime \prime}\left| P_{\Pi({\mathsf Y})} - P_{ {\mathsf Y}}\right| \label{eq:ass3a}\\
\leq&~
F_1 \zeta + \frac{F_2}{\sqrt n} \label{eq:ass3b}
\end{align}
where all constants $F$ are positive, and 
\begin{itemize}

\item \eqref{eq:ass3a} uses continuous differentiability of $P_{\mathsf X} \mapsto V(P_{\mathsf X}, P_{\mathsf Y})$ (in $\mathcal P_{\delta}^\star$) and $P_{\bar {\mathsf Y}} \mapsto V(P_{\mathsf X}, P_{\bar{\mathsf Y}})$ (at any $P_{\bar {\mathsf Y}}$ with $P_{\bar {\mathsf Y}} (\mathsf Y) > 0$ a.s.). 

\item \eqref{eq:ass3b} applies \eqref{eq:-minEuclid} and \eqref{eq:PYhat}. 

\end{itemize}

Theorem \ref{thm:minprob} is thereby applicable.

If $V(\beta) > 0$,
letting
\begin{align}
\gamma &= \frac 1 2 \log n \label{eq:-Ccgamma}\\
 \log M &= n C(\beta) - \sqrt{n V(\beta)}~\Qinv{\epsilon + \frac{ K + 1}{\sqrt n}} + \frac 1 2  \log n 
 \notag\\
 &
 + A \label{eq:-Cc},
\end{align}
where constant $K$ is the same as in \eqref{eq:minprob:1}, we apply Theorem \ref{thm:minprob}.\,\ref{thm:minprob:1} to conclude that the right side of \eqref{eq:-Ca} with minimization constrained to types in $\mathcal P_{\delta}^\star$ s lower bounded by $\epsilon$:
\begin{align}
 \min_{P_{\mathsf X} \in \mathcal P_\delta^\star} 
\Prob{\sum_{i = 1}^n W_i (P_{\mathsf X})
\leq \log M - \gamma - A } - \exp\left( -\gamma\right)  \geq \epsilon \label{eq:-Cb}.
\end{align}

If $V(\beta) = 0$, we fix $0 < \eta < 1 - \epsilon$ and let
 \begin{align}
\gamma &=  \log \frac 1 {\eta}, \\
 \log M &= n C(\beta) +  \left( \frac K {1 - \epsilon - \eta}\right)^{\frac 2 3} n^\frac 1 3 +  \log \frac 1 \eta \label{eq:-Cc0}, 
\end{align}
where $A$ is that in \eqref{eq:minprob:4}. Applying Theorem \ref{thm:minprob}.\ref{thm:minprob:4} with $\beta = \frac 1 6$, we conclude that \eqref{eq:-Cb} holds for the choice of $M$ in \eqref{eq:-Cc0} if $V(\beta) = 0$.

To evaluate the minimum over $\mathcal P \backslash \mathcal P_\delta^\star$ on the right side of \eqref{eq:-Ca}, 
define
\begin{equation}
 C(\beta) - \max_{P_{\mathsf X } \in \mathcal P \backslash \mathcal P_{\delta}^\star}
 \E{\jmath_{\mathsf X; \mathsf Y}(\mathsf X; \mathsf Y, \beta)} = 2 \Delta > 0 
 \end{equation}
 and observe
\begin{align}
&~
 D(P_{{\mathsf X}}, P_{\Pi( {\mathsf Y}) }) 
\notag\\
 =
 &~ 
 \E{\jmath_{ {\mathsf X}; {\mathsf Y}}({\mathsf X}; {\mathsf Y}, \beta)}
 +
 D({\mathsf Y}\| \Pi( {\mathsf Y}))
  + c|P_{\Pi({\mathsf Y}) } - P_{\mathsf Y^\star}|^2\\
 \leq
 &~ 
 \E{\jmath_{ {\mathsf X}; {\mathsf Y}}({\mathsf X}; {\mathsf Y}, \beta)}
 +
 D({\mathsf Y}\| \Pi( {\mathsf Y}))
  + 4 c   \label{eq:assca}\\
  \leq
  &~ 
  \E{\jmath_{ {\mathsf X}; {\mathsf Y}}({\mathsf X}; {\mathsf Y}, \beta)}
  + 
  \frac{ |\mathcal B|(|\mathcal B| - 1)\log e}{\sqrt {n c}} 
  +
 4c,
\label{eq:asscb} 
\end{align}
where 
\begin{itemize}

\item \eqref{eq:assca} holds because the Euclidean distance between two distributions satisfies
\begin{equation}
 |P_{\mathsf Y} - P_{\bar {\mathsf Y}}| \leq 2,
\end{equation}

\item  \eqref{eq:asscb} is due to \eqref{eq:PYhat}, \eqref{eq:divergence_ub}, and
\begin{equation}
 \min_{ \mathsf Y} \min_{\mathsf y \in \mathcal B} P_{ \Pi(\mathsf Y)}(\mathsf y) \geq \frac 1 {\sqrt{n c}} \label{eq:PYhatmin},
\end{equation}
which is a consequence of \eqref{eq:kset}.

\end{itemize}

Therefore, choosing $c < \frac \Delta 4$, we can ensure that for all $n$ large enough,
\begin{equation}
 C(\beta) -  \max_{P_{\mathsf X } \in \mathcal P \backslash \mathcal P_{\delta}^\star}
 D(P_{\mathsf X}, P_{\Pi(\mathsf Y)}) \geq \Delta > 0.
\end{equation}
Also, it is easy to show using \eqref{eq:PYhatmin} that there exists $a > 0$ such that
\begin{align}
V(P_{\mathsf X}, P_{\Pi(\mathsf Y)}) \leq a \log^2 n.
\end{align}

By Chebyshev's inequality, we have, for the choice of $\gamma$ in \eqref{eq:-Ccgamma} and $M$ in \eqref{eq:-Cc}, 
\begin{align}
 &~
 \max_{P_{\mathsf X} \in \mathcal P \backslash \mathcal P_\delta^\star}\Prob{\sum_{i = 1}^n W_i (P_{\mathsf X}) > \log M - \gamma - A }
 \notag\\
 \leq&~ \Prob{\sum_{i = 1}^n W_i (P_{\mathsf X}) -\E{W_i (P_{\mathsf X})}  >\frac{n \Delta}{2} } \label{eq:-Ccheb1}\\
 \leq&~ \frac {4 a } {\Delta^2} \frac{\log^2 n}{n}\label{eq:-Cd}.
\end{align}

Combining \eqref{eq:-Cb} and \eqref{eq:-Cd} concludes the proof.

 \section{Proof of the achievability part of Theorem \ref{thm:2orderCost}}
\label{appx:2orderCostA}
The proof consists of the asymptotic analysis of the following bound. 
\begin{thm}[Dependence Testing bound \cite{polyanskiy2010channel}]
There exists an $(M, \epsilon, \beta)$ code with 
\begin{equation}
\epsilon \leq \inf_{P_X}\E{\exp\left( - \left| \imath_{X; Y}(X; Y) - \log \frac{M - 1}{2}\right|^+\right) } \label{eq:DT}, 
\end{equation}
where the infimum is over all distributions supported on $\{x \in \mathcal X: \mathsf b(x) \leq \beta\}$.
\label{thm:DT}
\end{thm}

The following lemma will be instrumental. 
\begin{lemma}[\hspace{-.1mm}{\cite[Lemma 47]{polyanskiy2010channel}}]
Let $W_1, \ldots, W_n$ be independent, with $V_n > 0$ and $T_n < \infty$ where $V_n$ and $T_n$ are defined in \eqref{eq:BerryEsseenVn} and \eqref{eq:BerryEsseenTn}, respectively. Then for any $\gamma > 0$, 
\begin{align}
&~
 \E{\exp\left\{ - \sum_{i = 1}^n W_i \right\} 1 \left\{ \sum_{i = 1}^n W_i > \log \gamma \right\} }
\notag\\
 \leq &~ 
2 \left( \frac{\log 2}{\sqrt{2 \pi}} + \frac {2 T_n}{\sqrt {n V_n}}\right)\frac{1}{\gamma \sqrt {n V_n}}. 
\end{align}  
\label{lemma:polyanskiy}
\end{lemma}

Let $P_{X^n}$ be equiprobable on the set of sequences of type $P_{\hat {\mathsf X}^\star}$, where $P_{\hat {\mathsf X}^\star}$ is the minimum Euclidean distance approximation of $P_{\mathsf X^\star}$ formally defined in \eqref{eq:-Pi}. Let 
$P_{X^n} \to P_{Y^n | X^n} \to P_{Y^n}$,
$P_{\hat {\mathsf X}^\star} \to P_{\mathsf Y | \mathsf X} \to P_{ \hat {\mathsf Y}^\star}$,
and $P_{ \hat Y^{n \star}} = P_{ \hat {\mathsf Y}^\star} \times \ldots \times P_{ \hat {\mathsf Y}^\star}$.

The following lemma demonstrates that $P_{Y^n}$ is close to $P_{\hat Y^{n^\star}}$. 
\begin{lemma} Almost surely, for $n$ large enough and some constant $c$, 
\begin{equation}
\imath_{Y^n \| \hat Y^{n \star}} (Y^n) \leq 
\frac 1 2 \left( 
 \left| \mathrm{supp}\left(P_{{ \mathsf X}^\star}\right) \right| - 1\right) \log n + c \label{eq:lemmadivergence}
\end{equation}
\label{lemma:divergence}
\end{lemma}

\begin{proof}
For a vector $\mathbf k = (k_1, \ldots, k_{|\mathcal B|})$, denote the multinomial coefficient
\begin{equation}
{n \choose \mathbf k} = \frac{n!}{k_1! k_2! \ldots k_{|\mathcal B|}!}
\end{equation}
By Stirling's approximation, the number of sequences of type $P_{\hat {\mathsf X}^\star}$ satisfies, for $n$ large enough and some constant $c_1 > 0$ 
\begin{equation}
{n \choose n P_{\hat {\mathsf X}^\star}} \geq c_1 n^{- \frac 1 2 \left( 
 \left| \mathrm{supp}\left(P_{{ \mathsf X}^\star}\right) \right| - 1\right) } \exp \left( n H(\hat{\mathsf X}^\star )\right)
\end{equation}
On the other hand, for all $x^n$ of type $ P_{\hat X^{\star n}}$, 
\begin{equation}
 P_{\hat X^{\star n}} (x^n) =  \exp \left( - n H(\hat{\mathsf X}^\star )\right)
\end{equation}
Assume without loss of generality that all outputs in $\mathcal B$ are accessible, which implies that $P_{\mathsf Y^\star}(\mathsf y) > 0$ for all $\mathsf y \in \mathcal B$. Hence, the left side of \eqref{eq:lemmadivergence} is almost surely finite, and for all $y^n \in \mathcal Y^n$ with nonzero probability according to $P_{Y^n}$, 
\begin{align}
\frac{P_{Y^n}(y^n)}{P_{\hat Y^{n \star}}(y^n) } &=  \frac{{n \choose n P_{\hat {\mathsf X}^\star}}^{-1}\sum^\star P_{Y^n | X^n = x^n}(y^n)}
{ \sum_{x^n \in \mathcal A^n} P_{Y^n | X^n = x^n}(y^n) P_{\hat {X}^{n \star}}(x^n)}\\
&\leq  \frac{{n \choose n P_{\hat {\mathsf X}^\star}}^{-1} \sum^\star P_{Y^n | X^n = x^n}(y^n)}
{\sum^\star P_{Y^n | X^n = x^n}(y^n) P_{\hat {X}^{n \star}}(x^n)}\\
&=  \frac{{n \choose n P_{\hat {\mathsf X}^\star}}^{-1} \sum^\star P_{Y^n | X^n = x^n}(y^n)}
{\exp \left( - n H(\hat{\mathsf X}^\star )\right) \sum^\star P_{Y^n | X^n = x^n}(y^n)}\\
&={n \choose n P_{\hat {\mathsf X}^\star}}^{-1} \exp \left( n H(\hat{\mathsf X}^\star )\right) \\
&\leq c_1 n^{ \frac 1 2 \left( 
 \left| \mathrm{supp}\left(P_{{ \mathsf X}^\star}\right) \right| - 1\right) }, 
\end{align}
where we abbreviated
$
\sum^\star = \sum_{x^n \colon \mathrm{type}(x^n) = P_{\hat {\mathsf X}^\star}} 
$.
\end{proof}
 
We first consider the case $V(\beta) > 0$. For $c$ in \eqref{eq:lemmadivergence} and some $\gamma > 0$, let
\begin{align}
\log \frac{M - 1} 2&\triangleq S_n - \frac 1 2 \left( 
 \left| \mathrm{supp}\left(P_{{ \mathsf X}^\star}\right) \right| - 1\right) \log n - c \label{eq:-AlogM},\\
 S_n &\triangleq n D_n - \sqrt{n V_n} \Qinv{ \epsilon_n },\\
 \epsilon_n &\triangleq \epsilon - 2 \left( \frac{\log 2}{\sqrt{2 \pi}} + \frac {2 T_n}{\sqrt {n V_n}}\right) \frac{1}{\gamma \sqrt {n V_n}} - \frac{B_n}{\sqrt n},
\end{align}
where $D_n$ and $V_n$ are those in \eqref{eq:BerryEsseenDn} and \eqref{eq:BerryEsseenVn}, computed with $W_i = \imath_{\mathsf X; \hat{\mathsf Y}^\star}(x_i, Y_i)$, namely
\begin{align}
D_n &=  \E{\imath_{\mathsf X; \hat{\mathsf Y}^\star}(\hat {\mathsf X}^\star, \hat {\mathsf Y}^\star)}\\
V_n &= \Var{\imath_{\mathsf X; \hat{\mathsf Y}^\star}(\hat {\mathsf X}^\star, \hat {\mathsf Y}^\star) | \hat{\mathsf X}^\star}
\end{align}
Since the functions $P_{\mathsf X} \mapsto \E{\imath_{\mathsf X; \mathsf Y}(\mathsf X, \mathsf Y)}$ and $P_{\mathsf X} \mapsto \Var{\imath_{\mathsf X; \mathsf Y}(\mathsf X, \mathsf Y)|\mathsf X}$ are continuously differentiable in a neighborhood of $P_{\mathsf X^\star}$ in which $P_{\mathsf Y}(\mathsf Y) > 0$ a.s., there exist constants $L_1 \geq 0$, $F_1 \geq 0$ such that
\begin{align}
|D_n - C(\beta)| &\leq  L_1 |P_{\hat {\mathsf X}^\star} -  P_{{\mathsf X}^\star}|, \\
|V_n - V(\beta)| &\leq  F_1 |P_{\hat {\mathsf X}^\star} -  P_{ {\mathsf X}^\star}|,\label{eq:-AF1}
\end{align}
where we used \eqref{eq:varinfodensity}. Applying \eqref{eq:-minEuclid}, we observe that the choice of $\log M$ in \eqref{eq:-AlogM} satisfies \eqref{eq:2ordercost}, \eqref{eq:remainderCostA}. Therefore, to prove the claim we need to show that the right side of \eqref{eq:DT} with the choice of $M$ in \eqref{eq:-AlogM} is upper bounded by $\epsilon$. 

Weakening \eqref{eq:DT} by choosing $P_{X^n}$ equiprobable on the set of sequences of type $P_{\hat {\mathsf X}^\star}$, as above, we infer that an $(M, \epsilon^\prime, \beta)$ code exists with 

\begin{align}
 \epsilon^\prime &\leq \E{\exp\left( - \left| \imath_{X^n; Y^n}(X^n; Y^n) - \log \frac{M - 1}{2}\right|^+\right) }\\
 &=   \mathbb E \Big[ \exp \Big( - \Big| \sum_{i = 1}^n \imath_{\mathsf X; \hat{\mathsf  Y}^{\star} }(X_i; Y_i) 
 - \imath_{Y^n \| \hat Y^{n \star}} (Y^n)  
 \notag\\
 &
 - \log \frac{M - 1}{2}\Big|^+ \Big)  \Big] \\
 &\leq  \E{\exp\left( - \left| \sum_{i = 1}^n \imath_{\mathsf X; \hat{\mathsf  Y}^{\star} }(X_i; Y_i) 
 - S_n \right|^+\right)} \label{eq:-Aa}\\
 &= \E{\exp\left( - \left| \sum_{i = 1}^n \imath_{\mathsf X; \hat{\mathsf  Y}^{\star} }(x_i; Y_i) 
 - S_n \right|^+\right) } \label{eq:-Ab}\\
 &\leq \exp \left(S_n\right)  \cdot
\notag\\ 
&
 \E{\exp\left( - \sum_{i = 1}^n \imath_{\mathsf X; \hat{\mathsf  Y}^{\star} }(x_i; Y_i)  \right) 1\left\{\sum_{i = 1}^n \imath_{\mathsf X; \hat{\mathsf  Y}^{\star} }(x_i; Y_i) > S_n \right\} } \notag\\
 &+ \Prob{\sum_{i = 1}^n \imath_{\mathsf X; \hat{\mathsf  Y}^{\star} }(x_i; Y_i) \leq S_n }  \\
 &\leq \epsilon \label{eq:-Ac},
\end{align}
where 
\begin{itemize}
\item \eqref{eq:-Aa} applies Lemma \ref{lemma:divergence} and substitutes \eqref{eq:-AlogM};
\item \eqref{eq:-Ab} holds for any choice of $x^n$ of type $P_{\hat {\mathsf X}^\star}$ because the (conditional on $X^n = x^n$) distribution of $\imath_{X^n; \hat Y^{n \star} }(x^n; Y^n) = \sum_{i = 1}^n \imath_{\mathsf X; \hat{\mathsf  Y}^{\star} }(x_i; Y_i)$ depends the choice of $x^n$ only through its type;
\item \eqref{eq:-Ac} upper-bounds the first term using Lemma \ref{lemma:polyanskiy}, and the second term using Theorem \ref{thm:Berry-Esseen}. 
\end{itemize}

If $V(\beta) = 0$, let $S_n$ in \eqref{eq:-AlogM} be
\begin{equation}
S_n = n D_n - 2 \gamma,
\end{equation}
and let $\gamma > 0$ be the solution to 
\begin{equation}
\exp(-\gamma) + \frac{F_1 \sqrt{|\mathcal A| (|\mathcal A| - 1)}}{\gamma^2} = \epsilon \label{eq:-Ad},
\end{equation}
where $F_1$ is that in \eqref{eq:-AF1}. Note that such solution exists because the function in the left side of \eqref{eq:-Ad} is continuous on $(0, \infty)$, unbounded as $\gamma \to 0$ and vanishing as $\gamma \to \infty$.  The reasoning up to \eqref{eq:-Ab} still applies, at which point we upper-bound the right-side of \eqref{eq:-Ab} in the following way:
\begin{align}
\epsilon^\prime &\leq \exp\left(-\gamma \right) \Prob{ \sum_{i = 1}^n \imath_{\mathsf X; \hat{\mathsf  Y}^{\star} }(x_i; Y_i) > S_n + \gamma } 
\notag\\
&
+ \Prob{ \sum_{i = 1}^n \imath_{\mathsf X; \hat{\mathsf  Y}^{\star} }(x_i; Y_i) \leq S_n + \gamma } \\
&\leq \exp\left(-\gamma \right) + \frac{n V_n}{\gamma^2} \label{eq:-Ae}\\
&\leq \epsilon \label{eq:-Af},
\end{align}
where 
\begin{itemize}
\item  \eqref{eq:-Ae} upper-bounds the second probability using Chebyshev's inequality; 
\item \eqref{eq:-Af} uses $V(\alpha) = 0$, \eqref{eq:-minEuclid} and \eqref{eq:-AF1}. 
\end{itemize}


\section{Proof of Theorem \ref{thm:2orderCost} under the assumptions of Remark \ref{remark:continuous}}
\label{appx:continuous}
Under assumption \eqref{item:costunbounded}, every $(n, M, \epsilon, \beta)$ code with a maximal cost constraint can be converted to an $(n+1, M, \epsilon, \beta)$ code with an equal cost constraint (i.e. equality in \eqref{eq:costmax} is requested) by appending to each codeword a coordinate $x_{n+1}$ with
\begin{equation}
 \mathsf b(x_{n+1}) = \beta  - \sum_{i = 1}^n \mathsf b(x_i) \label{eq:-conta}.
\end{equation}
Since $\sum_{i = 1}^n \mathsf b(x_i) \leq \beta n$, the right side of \eqref{eq:-conta} is no smaller than $\beta$, and so by assumption \eqref{item:costunbounded} a coordinate $x_{n+1}$ satisfying \eqref{eq:-conta} can be found. It follows that
\begin{equation}
M^\star_{\mathrm{eq}}(n, \epsilon, \beta) 
\leq  M^\star_{\mathrm{max}}(n, \epsilon, \beta)
\leq M^\star_{\mathrm{eq}}(n+1, \epsilon, \beta),
\end{equation}
where the subscript specifies the nature of the cost constraint. We thus may focus only on the codes with equal cost constraint. The capacity-cost function can be expressed as \eqref{eq:capacitycostsym} due to \eqref{eq:C(b)jstarconditional}. 
The converse part now follows by invoking \eqref{eq:Ccosta}  with $P_{\bar Y^n} = P_{\mathsf Y^\star} \times \ldots \times P_{\mathsf Y^\star}$ and $\gamma = \frac 1 2 \log n$. A simple application of the Berry-Esseen bound (Theorem \ref{thm:Berry-Esseen}) using assumption \eqref{item:thirdmoment} leads to the desired result. 

To show the achievability part, we follow the proof in Appendix \ref{appx:2orderCostA}, drawing the codewords from $P_{X^n}$ appearing in assumption \eqref{item:idivbounded},  replacing all minimum distance approximations by the true distributions, and replacing the right side of \eqref{eq:lemmadivergence} by $f_n$.

\section{Dispersion-cost function of an additive exponential channel}
\label{appx:exp}

As shown in \cite{verdu1996exponential}, the capacity-cost function is given by \eqref{eq:Cexp}, and $\mathsf Y^\star$ is exponential with mean $1 + \beta$, i.e.
\begin{equation}
dP_{\mathsf Y^\star}(y) = \frac 1 {1 + \beta}  e^{ - \frac y {1 + \beta}} dy,
\end{equation}
which leads to the expression for $\mathsf b$-tilted information density in \eqref{eq:btiltedexp}. Conditions \eqref{item:costunbounded}--\eqref{item:thirdmoment} in Remark \ref{remark:continuous} are clearly satisfied. To verify condition \eqref{item:idivbounded}, let $P_{X^n}$ be uniform on the $(n-1)$-simplex $\{x^n \in \mathbb R_+^n \colon \sum_{i = 1}^n x_i = n \beta\}$. Then, the distribution of $Y^n = X^n + N^n$, where $N^n$ is a vector of i.i.d. exponential components with means $1$, is a function of $\sum_{i = 1}^n N_i$ only. Since the same holds for $Y^{n \star}$, the log-likelihood ratio $\imath_{Y^n \| Y^{n \star}}(y^n)$ is also a function of $\sum_{i = 1}^n y_i$ only. Now, the sum of $n$ exponentially distributed random variables with mean $a$ has Erlang distribution, whose pdf is $\frac {t^{n - 1} e^{-t/a}}{a^n (n - 1)!} dt$, so (assuming natural logarithms for ease of computation)
\begin{align}
\imath_{Y^n \| Y^{n \star}}(y^n) &= L\left(\sum_{i = 1}^n y_i, n\right),\\
L(t, n) &\triangleq n \beta - \frac {\beta}{1 + \beta} t + n \log_e (1 + \beta) 
\notag\\
&
+ (n - 1) \log_e \left( 1 - \frac {n \beta}{t}\right). 
\end{align}
A direct algebraic computation shows that for each $n$, the maximum of $L(\cdot, n)$ is achieved at 
\begin{equation}
 t^\star(n) \triangleq \frac 1 2 \left( n \beta + \sqrt n \sqrt{n \beta^2 + 4 n (1 + \beta) - 4(1 + \beta) }\right).
\end{equation}
Another computation verifies that $L(t^\star(n), n)$ is monotonically decreasing in $n$, so 
\begin{align}
\max_{n, t} L(t, n)  &= L(t^{\star}(1), 1)\\
&= \frac {\beta}{1 + \beta} + \log_e (1 + \beta), 
\end{align}
i.e. $\imath_{Y^n \| Y^{n \star}}(y^n)$ is bounded by a constant, and condition \eqref{item:idivbounded} is satisfied. 

\bibliographystyle{IEEEtran}
\bibliography{../../ratedistortion}
\end{document}